\documentclass[%
 reprint,
superscriptaddress,
nofootinbib,
amsmath,amssymb,aps,floatfix,
prd,
]{revtex4-2}
\usepackage{color}
\usepackage{graphicx}
\usepackage{graphics}
\usepackage{caption}
\usepackage{amsthm,lmodern}
\usepackage{bigints}
\usepackage{multirow}
\usepackage{subcaption}
\usepackage{dcolumn}
\usepackage{bm}
\RequirePackage{mathptmx}   
\usepackage[normalem]{ulem}
\DeclareMathAlphabet{\mathcal}{OMS}{cmsy}{m}{n}
\RequirePackage{amsmath,amssymb}
\RequirePackage{bm}

\usepackage{mathtools}

\usepackage{adjustbox}
\RequirePackage{bbm}
\usepackage[ruled]{algorithm2e}
\RequirePackage{flushend}
\RequirePackage{natbib}
\usepackage[colorlinks,citecolor=blue,urlcolor=blue,linkcolor=blue]{hyperref}
\RequirePackage[usenames, dvipsnames]{xcolor}
\RequirePackage{tikz}
\RequirePackage[percent]{overpic}
\RequirePackage{setspace}
\usetikzlibrary{shapes.geometric, arrows}
\usepackage{natbib}
\newcommand{\rpm}{\sbox0{$1$}\sbox2{$\scriptstyle\pm$}
  \raise\dimexpr(\ht0-\ht2)/2\relax\box2 }

\usepackage{orcidlink}
\usepackage{verbatim}

\newtheorem{theorem}{Theorem}
\newtheorem{lemma}[theorem]{Lemma}
\newtheorem{proposition}{Proposition}

\newcommand{\Rel}{\operatorname{Re}}
\newcommand{\Img}{\operatorname{Im}}

\newcommand{\fb}{{f_b}}
\newcommand{\fc}{{f}}

\newcommand{\Ahat}{\hat{\bm{A}}}
\newcommand{\Atilde}{\tilde{\bm{A}}}

\newcommand{\bmtheta}{\bm{\theta}}
\newcommand{\bmthetao}{\bm{\theta}_0}

\newcommand{\thetahat}{\hat{\bmtheta}}

\newcommand{\residuals}{ \hat{\bm{\epsilon}}}

\newcommand{\errors}{\bm{\epsilon}}

\newcommand{\argmin}{\mathop{\rm arg\,min}}
\newcommand{\inner}[1]{\left\langle #1 \right\rangle}

\newcommand{\inP}{ \stackrel{\mathbb{P}}{\rightarrow}}

\begin{document}


\title{On Validating Angular Power Spectral Models for the Stochastic Gravitational-Wave Background Without Distributional Assumptions}
\thanks{This paper has been submitted in conjunction with its companion \citet{PRL}}%
\author{Xiangyu Zhang\,\orcidlink{0000-0003-3684-0370}}
\affiliation{School of Statistics, University of Minnesota, Minneapolis, MN,  USA}
\author{Erik Floden\,\orcidlink{0000-0002-1701-7461}}
\affiliation{School of Physics and Astronomy, University of Minnesota, Minneapolis, MN,  USA}
\author{Hongru Zhao\,\orcidlink{0000-0002-1753-4709}}
\affiliation{School of Statistics, University of Minnesota, Minneapolis, MN,  USA}
\author{Sara Algeri\,\orcidlink{0000-0001-7366-3866}}
\email{salgeri@umn.edu}
\affiliation{School of Statistics, University of Minnesota, Minneapolis, MN,  USA}
\author{Galin Jones\,\orcidlink{0000-0002-6869-6855}}
\affiliation{School of Statistics, University of Minnesota, Minneapolis, MN,  USA}
\author{Vuk Mandic\,\orcidlink{0000-0001-6333-8621}}
\affiliation{School of Physics and Astronomy, University of Minnesota, Minneapolis, MN,  USA}
\author{Jesse Miller\,\orcidlink{0009-0005-9465-7461}}
\affiliation{School of Statistics, University of Minnesota, Minneapolis, MN,  USA}
\date{\today}
\begin{abstract}
It is demonstrated that estimators of the angular power spectrum commonly used for the stochastic gravitational-wave background (SGWB) lack a closed-form analytical expression for the likelihood function and, typically, cannot be accurately approximated by a Gaussian likelihood. Nevertheless, a robust statistical analysis can be performed by extending the framework outlined in \cite{PRL} to enable the estimation and testing of angular power spectral models for the SGWB without specifying distributional assumptions. Here, the technical aspects of the method are discussed in detail. Moreover, a new, consistent estimator for the covariance of the angular power spectrum is derived. The proposed approach is applied to data from the third observing run (O3) of Advanced LIGO and Advanced Virgo.

\end{abstract}


\maketitle

\section{\label{sec:1}Introduction}
Estimating and testing models for the angular power spectrum is required in many astrophysical and cosmological settings, including the cosmic microwave background \cite{planck2020}, galaxy sky surveys \cite{SDSS_DR16}, and weak gravitational lensing surveys \cite{DESY3LensingMap}. The focus here is on the stochastic gravitational-wave background (SGWB) resulting from the superposition of gravitational-wave (GW) signals that are too weak or too numerous to be resolved individually. While the SGWB and its anisotropies have not yet been detected by second-generation GW detectors such as the Advanced LIGO \cite{Aasi_2015} and Advanced Virgo \cite{Acernese_2015}, upper limits have been placed on the GW energy density of the background as well as its angular power spectra \cite{O3stochdir, o3stochiso, o4isotropic, o4aniso}.

Anisotropies in the SGWB are expected to arise from a variety of cosmological and astrophysical phenomena \cite{Contaldi_2017, jenkins2019, Jenkins_2019_2, Pitrou_2020}. They can provide information on the astrophysical distribution of GW sources within galaxies, formation of large-scale structure in the universe, and different GW production mechanisms \cite{Cusin:2017fwz, Cusin:2017mjm, Cusin_2018_2, Cusin:2019jpv}. GWs generated by cosmological phase transitions necessarily have anisotropies analogous to those of the cosmic microwave background and may provide information regarding primordial inhomogeneities during an early inflationary era \cite{Geller, Bertacca_2020}. GW signals emitted by neutron stars within our galaxy can trace out the galactic plane \cite{Talukder_2014, Mazumder_2014}, while cosmic strings may cause anisotropies at small angular scales \cite{Jenkins_2018}.

Theoretical models of SGWB anisotropy typically do not predict the realization of the anisotropy on the sky; instead, they predict the distribution of the SGWB power across different angular scales described by the SGWB angular power spectrum. It is, therefore, critical to measure the SGWB angular power spectrum in order to establish a comparison with theoretical models. In practice, one starts by cross-correlating the strain time-series data measured by two GW detectors located at two different locations on Earth. Under the assumption of a stationary SGWB, one can then leverage the rotating antenna pattern of the detector pair (induced by the rotation of the Earth) to estimate the SGWB anisotropy across the sky from the GW data~\cite{allenromano, Thrane:2009fp, O3stochdir}. Expanding the estimated anisotropy in terms of spherical harmonics, and then squaring and summing the expansion coefficients, results in estimators of the SGWB angular power spectrum (see Sec. \ref{sec:GW}). These estimators can then be directly compared with their theoretical counterparts.

Estimators of the SGWB anisotropy, such as coefficients in the spherical harmonic expansion, are obtained by averaging over many time segments of the strain data and are, therefore, asymptotically multivariate Gaussian distributed. This feature has been used in past parameter estimation analyses whose objective was to estimate individual spherical harmonic modes \cite{Tsukada_2023} or to estimate a specific anisotropy realization on the sky, such as the galactic plane \cite{Agarwal_MilkyWay}. Estimators of the angular power spectrum are obtained by squaring the SGWB anisotropy estimators; hence, they are not Gaussian distributed. Nevertheless, some studies have assumed Gaussianity of the angular power spectrum \cite{Agarwal_2023}, but such an assumption has been shown to lead to biases.



As shown in Appendix~\ref{generalized_chi}, the estimators of the angular power spectrum follow a generalized $\chi^2$ distribution, which cannot be accurately approximated by the Gaussian distribution. 
Moreover, using the generalized $\chi^2$ likelihood is nontrivial because, in general, 
{\color{black} the corresponding probability density function} does not have a closed-form expression. Alternatively, one could devise a scheme that does not require any assumptions regarding the angular power spectrum's underlying distribution. The latter approach is used here to implement and adapt the statistical procedure outlined in our companion paper \cite{PRL} to the context of the SGWB.

The remainder of the manuscript is organized as follows. Sec. \ref{sec:GW} describes the search methods used to estimate the angular power spectra from GW detector data. Sec. \ref{sec:III} presents a comprehensive explanation of how the distribution-free statistical framework outlined in our companion paper \cite{PRL} can be tailored here to the specifics of an SGWB analysis. Sec. \ref{simulations} investigates the statistical properties of the proposed procedure via simulation studies. In Sec. \ref{sec:application}, the method is applied to GW data from Advanced LIGO and Advanced Virgo’s third observing run (O3), comparing it with a theoretical model for the angular power spectrum arising from the superposition of compact binary coalescences \cite{Cusin:2017mjm, Cusin:2019jpv, yang2023}. Our results demonstrate that this framework can reliably estimate the unknown parameters of the hypothesized models for the angular power spectrum and provides a powerful and computationally efficient goodness-of-fit test to assess the validity of such models.

\section{\label{sec:GW} Spherical Harmonics Search Methods for the SGWB}

\textcolor{black}{Data from ground-based interferometric gravitational-wave} detectors may be used to generate estimators of the SGWB angular power spectrum. Each GW detector, labeled by $I = 1,2$, produces a strain time series given by
\begin{align}
s_{I}(t) = h_{I}(t) + n_{I}(t),
\end{align}
with $n_I(t)$ denoting the detector noise and $h_I(t)$ is the detector response to a GW signal. The GW signal can be decomposed in terms of GW modes $h_A(f,\hat{\Omega})$ defined by their polarization $A$, frequency $f$, and propagation direction $\hat{\Omega}$:
\begin{align}
h_I(t) = \int^{\infty}_{-\infty} \mathrm{d}f \int_{S^2} \mathrm{d}\hat{\Omega}h_A(f,\hat{\Omega})F_I^A(\hat{\Omega},t)e^{i2\pi f(t-\hat{\Omega}\cdot \vec{x}_I(t)/c)}.
\end{align}
Here, $F_I^A(\hat{\Omega},t)$ is the time-dependent detector response function, $\vec{x}_I(t)$ gives the time-dependent location of interferometer $I$, and $S^2$ denotes integration over the 2-sphere. If $\hat{X}_I(t)$ and $\hat{Y}_I(t)$ are unit vectors pointing along the arms of detector $I$, the detector response function is
\begin{align}\label{eq:response_function}
F_I^A(\hat{\Omega},t) = \frac{1}{2}\left[\hat{X}^a_I(t)\hat{X}^b_I(t) - \hat{Y}^a_I(t)\hat{Y}^b_I(t)\right] e^{A}_{ab}(\hat{\Omega}),
\end{align}
in which we assume Einstein summation over the spatial indices $a,b$. The gravitational-wave polarization tensors $e^{A}_{ab}(\hat{\Omega})$ are defined as in \cite{Thrane:2009fp}.
\textcolor{black}{Note the absence of frequency dependence in Eq.~\eqref{eq:response_function} holds only when working in the small antenna limit \cite{RomanoCornish}, a limit which we assume throughout this work.} We assume the SGWB is unpolarized, Gaussian, and stationary. Additionally, by assuming perturbations to the spacetime metric are real, the Fourier coefficients $h_A(f,\hat{\Omega})$ satisfy $h_A(-f,\hat{\Omega}) = h_A^*(f,\hat{\Omega})$, where $^*$ denotes complex conjugation. \textcolor{black}{Without loss of generality, we also assume they} have zero means, $\langle h_A(f,\hat{\Omega})\rangle = 0$ \textcolor{black}{\cite{RomanoCornish}}, and the expected value of their two-point correlation is 
\begin{align}\label{eq:2point}
\langle h_{A}^{*} (f,\hat\Omega) \, h_{A'} (f',\hat\Omega')\rangle = \frac{1}{4}\, \mathcal{P}(f,\hat\Omega)\,\delta_{AA'}\,\delta(f-f')\,\delta^2(\hat\Omega,\hat\Omega'),\,
\end{align}
where ${\mathcal{P}}(f,\hat{\Omega})$ gives the spectral and angular distribution of the SGWB \textcolor{black}{power} \cite{Thrane:2009fp}, $\delta_{AA'}$ denotes the Kronecker delta function in polarization, and the remaining $\delta$'s denote Dirac delta functions in frequency and angular position on the sky. \textcolor{black}{The factor of $\frac{1}{4}$ in Eq.~\eqref{eq:2point} is comprised of one factor of $\frac{1}{2}$ which accounts for our considering only the one-sided ($f>0$) spectral density and another which arises from our averaging over the two GW polarizations \cite{Thrane:2009fp, PhysRevD.110.043031}.} We also 
assume that ${\mathcal{P}}(f,\hat{\Omega})$ can be factored into its separate spectral and angular components:
    \begin{align}
    \label{eq:p_f_theta}
     \mathcal{P}(f, \hat{\Omega}) = H(f) \, \mathcal{P}(\hat{\Omega}),
    \end{align}
where $H(f)$ is a dimensionless quantity assumed to be a power law both for its simplicity and its ability to approximate most interesting SGWB models \cite{o1iso, allenromano, s4iso}, 
     \begin{align}
     H(f) =\left(\frac{f}{f_{\rm ref}}\right)^{\alpha -3}.
     \end{align}
\textcolor{black}{The factorization in Eq.~\eqref{eq:p_f_theta} is performed under the assumption that the GW power is constant over each frequency band we consider, and it is used in other SGWB analyses, including broadband searches with bandwidths over 1700 Hz \cite{Thrane:2009fp,o4aniso}.} We set the reference frequency $f_{\rm ref} = 25$ Hz \textcolor{black}{as is used in recent SGWB analyses \cite{O3stochdir,o4aniso}, and}  the spectral index $\alpha$ may be assigned a value depending on the type of SGWB model being considered. For the astrophysical SGWB resulting from merger events between compact objects such as black holes and neutron stars, $\alpha = 2/3$ \cite{O3stochdir}.

Next, \textcolor{black}{we decompose} the angular distribution of the background $\mathcal{P}(\hat{\Omega})$ into its spherical harmonic components $Y_{\ell m}$ and coefficients $a_{\ell m}$:
\begin{align}
\mathcal{P}(\hat\Omega) = \sum_{\ell=0}^{\ell_{\rm max}} \sum_{m=-\ell}^{\ell}  a_{\ell m}\,Y_{\ell m}(\hat\Omega).
\end{align} 
Note that these definitions assume that each GW frequency bin has the same anisotropy. This assumption can be relaxed, and the analysis can be pursued in separate frequency bins (or bands) $\fb$. In this case, the spherical harmonic coefficients will become frequency dependent and we denote them as $a_{\fb, \ell m}$ \textcolor{black}{\cite{PhysRevD.108.023011, yang2023}}. 

The value of $\ell_{\rm max}$ determines the highest-order spherical harmonic modes that are considered in the analysis. This value may be chosen differently for different spectral models of the SGWB, and while $\ell_{\rm max}=4$ is typically used when $\alpha = 2/3$ \cite{O3stochdir}, $\ell_{\rm max} = 8$ is used here. This choice of $\ell_{\rm max}$ is more appropriate for the frequency range over which we will evaluate the angular power spectrum. It has been shown that while $\ell_{\rm max}=4$ is an optimal choice for making a GW detection of the chosen SGWB model, it corresponds to an angular resolution that is lower than what is resolvable by the Advanced LIGO-Virgo detectors \cite{floden2022}, which justifies considering spherical harmonic modes corresponding to a higher angular resolution.

Next, apply a short-term Fourier transform of the time series output $s_I(t)$ to obtain
	\begin{align}
    \tilde{s}_I(\textcolor{black}{t;f}) \equiv \int^{t+\tau / 2}_{t - \tau / 2} dt' e^{-i 2 \pi f t'} s_I(t').
	\end{align}
Here, $\tau$ is chosen to be 192 seconds, which is long enough to be much greater than the light travel time between detectors but short enough so as to minimize effects due to changes in the detectors' response functions induced by the Earth rotatation \cite{Thrane:2009fp, O3stochdir}. The cross-correlation between detectors $I$ and $J$ is
	\begin{align}
		\hat{C}(t;f)=\frac{2}{\tau}\tilde{s}_I^*(t;f)\tilde{s}_J(t;f),
	\end{align}
where, to simplify the notation, the $I,J$ indices corresponding to the two detectors from the cross-spectra terms have been dropped. Assuming that the noise of the two detectors is uncorrelated \cite{Thrane:2009fp}\footnote{{\color{black}While some sources of correlated noise such as Schumann resonances do exist, they do not affect searches for the SGWB given the current sensitivities of Advanced LIGO and Advanced Virgo \cite{o3stochiso,Janssens:2023anf,Venikoudis:2024zmk}.}}, the expectation value of the cross-correlation spectrum can be expressed as
	\begin{align}
		\langle \hat{C}(t;f)\rangle=\sum_{\ell=0}^{\ell_{\rm max}} \sum_{m=-\ell}^{\ell}H(f)\gamma_{\ell m}(t;f){a}_{\fb,\ell m},
	\end{align}
where frequency $f$ belongs to the band $f_b$.
The overlap reduction function, $\gamma_{\ell m}(t;f)$, is a geometric factor that accounts for the relative orientation and separation of the two detectors in a given pair \cite{Thrane:2009fp}. The covariance matrix of the cross-correlation spectra is given by
\begin{align}
    N_{ft,f't'}= \langle \hat{C}(t;f)\hat{C}^*(t';f')\rangle - \langle \hat{C}(t;f)\rangle \langle \hat{C}^*(t';f')\rangle  \nonumber \\ 
    \approx \delta(t-t')\delta(f-f')P_I(t;f)P_J(t;f), \label{approx}
    \end{align}
Here, $P_I(t;f)$ and $P_J(t;f)$ correspond to the one-sided power spectral densities of the two GW detectors' outputs. Assuming the present GW signal powers (both cross- and auto-correlated) are much  {\color{black}smaller} than the detector noise power, we obtain the approximation in Eq. \eqref{approx}. For a given set of $a_{\fb,\ell m}$, the likelihood function for the cross-correlation spectra in the frequency band $f_b$ is proportional to
\begin{equation}\label{equ:gauss_density}
\begin{aligned}
 & \exp\Bigl(-\sum_{\fc,t} [\hat{C}(t;\fc)- H(\fc) \, \sum_{\ell m}\gamma_{\ell m}(t;\fc) \, {a}_{\fb,\ell m}]^*
\\
& \quad N^{-1}_{\fc t,\fc t} [\hat{C}(t;\fc)- H(\fc) \, \sum_{\ell m}\gamma_{\ell m}(t;\fc) \, {a}_{\fb,\ell m}]\Bigl)
\end{aligned}
\end{equation}
and the summation is over $\fc \in \fb$ and $t \in [0,T)$. The maximizers of this likelihood function have the form 
	\begin{align}\label{eq:p0}
		\hat{a}_{\fb,\ell m}=\sum_{\ell' m'} \big(\Gamma_{\fb}^{-1}\big)_{\ell m, \ell' m'}\hat{X}_{\fb,\ell' m'}
	\end{align}
where 
	\begin{align}\label{eq:dirty}
		\hat{X}_{\fb,\ell m}=\sum_{\fc\in \fb,t\in [0,T) }\gamma_{\ell m}^* (t;\fc)\frac{H(\fc)}{P_I(t;\fc) P_J(t;\fc)}\hat{C}(t;\fc)
	\end{align}
	\begin{align}\label{eq:fisher}
		\Gamma_{\fb,\ell m ,\ell' m'} = \sum_{\fc\in \fb,t\in [0,T) } \gamma_{\ell m}^*(t;\fc) \frac{H^2(\fc)}{P_I(t;\fc) P_J(t;\fc)}\gamma_{\ell' m'}(t;\fc).
	\end{align}

\textcolor{black}{Eq. \eqref{eq:dirty} defines a map of the gravitational-wave background convolved with the antenna pattern of the detector pair, the so-called dirty map. The covariance matrix of the dirty map is given in  Eq. \eqref{eq:fisher} and is referred to as the Fisher matrix. Since Eq. \eqref{eq:p0} represents the deconvolution of the GW signal from the detector response, it is the so-called clean map. The covariance matrix of the clean map is obtained by inverting the covariance matrix of the dirty map. Detector insensitivity to certain sky directions manifests in small eigenvalues in the Fisher matrix that complicate its inversion. Such regularization is typically implemented via singular value decomposition to mitigate small singular values arising from directions in the sky to which the detector network is insensitive \citep{Thrane:2009fp,O3stochdir}. However, this procedure introduces a bias that propagates the estimators and covariance of the angular power spectrum in non-trivial ways \citep{yang2023}. To avoid this bias, the Fisher matrix is not regularized in this study. Consequently, a lack of regularization will give increased uncertainties in estimators of the clean map and angular power spectrum \cite{o4aniso}.}

 Note that if Eq. (\ref{equ:gauss_density}) holds,  each $\hat{a}_{\fb,\ell m}$ is a linear combination of complex Gaussian variables; thus, they also
follow a complex Gaussian distribution with mean 
$\langle \hat{a}_{\fb,\ell m} \rangle={a}_{\fb,\ell m}$ (discussed further below) 
and covariance $\operatorname{Cov} (\hat{a}_{\fb,\ell m},\hat{a}_{\fb,\ell' m'} )=(\Gamma^{-1}_{\fb})_{\ell m,\ell' m'}.$ Indeed the same is true, at least approximately, even if Eq. (\ref{equ:gauss_density}) does not hold, but the summation over $t$ in Eq.~\eqref{eq:dirty} is taken over a large number of segments (as is the case with the Advanced LIGO and Advanced Virgo observing runs). \textcolor{black}{This is true due to the central limit theorem as well as the assumption that the SGWB signal is weak relative to detector noise \cite{allenromano,RomanoCornish,Regimbau_2011}.} 

The estimates $\hat{a}_{\fb,\ell m}$ are evaluated by summing over discrete bins in the frequency band $\fb$. Broadband analyses perform this sum over the entire available range of frequencies to which the detectors are most sensitive \cite{O3stochdir}. However, in the analysis presented here, we obtain multiple samples of $\hat{a}_{\fb,\ell m}$ corresponding to sums over different frequency bands. The subscript $\fb$ used in Eq. \eqref{equ:gauss_density} and onwards denotes the frequency range over which each set of $\hat{a}_{\fb,\ell m}$, $\Gamma_{\fb,\ell m,\ell' m'}$, and $\hat{X}_{\fb,\ell m}$ is evaluated. Similar notation is employed throughout the manuscript whenever a quantity is evaluated over some discrete range of frequencies rather than considering $\fb$ to be a continuous variable.

From a statistical perspective, maximum likelihood is fundamentally a method for estimating deterministic quantities. This means the $\hat{a}_{\fb,\ell m}$ in Eq.~\eqref{eq:p0} serve as maximum likelihood estimators of the spherical harmonic coefficients, $a_{\fb,\ell m}$, that uniquely characterize the SGWB of the only realization of the universe available. Therefore, similarly to what is implicitly done in many statistical analyses, the $a_{\fb,\ell m}$ are not treated as random variables \citep[e.g.,][]{Thrane:2009fp, abbott2017directional, abbott2019directional, O3stochdir}. This is a reasonable assumption in practical scenarios, and treating the $a_{\fb,\ell m}$ as fixed avoids the risk of incorrectly specifying their underlying distribution, which can substantially bias the resulting inference. 

The proposed statistical methodology requires that multiple replicates of $\hat{a}_{\fb,\ell m}$ are available. To ensure that is the case, we partition the total observation time $T$ into $S$ segments, each of duration $\Delta T$, so that $T=\Delta T \cdot S.$ \textcolor{black}{We also divide the frequency range into $B$ discrete frequency bins $f_b$.} 

For each \(s = 1, \dots, S\), define \(\hat{a}_{\fb,s,\ell m}\), \(\hat{X}_{\fb,s,\ell m}\), and \(\Gamma_{\fb,s,\ell m,\ell' m'}\) using Eqs.~\eqref{eq:p0}-\eqref{eq:fisher}, but sum over $t$ in the interval 
\(\bigl[(s-1)\Delta T,\,s \Delta T\bigr)\) instead of \([0,T)\). It can be easily verified that $\hat{a}_{\fb,s,\ell m}$ are independent (across $s$) complex Gaussian variables, each with mean \(\langle \hat{a}_{\fb,s,\ell m} \rangle = a_{\fb,\ell m}\) and covariance 
$ {\rm Cov}(\hat{a}_{\fb,s,\ell m}, \hat{a}_{\fb,s,\ell' m'}) 
= (\Gamma_{\fb,s}^{-1})_{\ell m, \ell' m'}$. Because the detector sensitivity varies over time -- due to factors such as transient environmental noise, detector maintenance, and modifications -- we do not assume that ${\Gamma}_{\fb,s}^{-1}$ is identical across different time segments $s.$

Since the detector network consists of two Advanced LIGO detectors and one Advanced Virgo detector, there are three pairs of detectors to consider, or three `baselines'. Each baseline has its own dirty map and Fisher matrix, and these may be summed together to obtain one overall dirty map and one overall Fisher matrix corresponding to the entire detector network \cite{Thrane:2009fp}. 

Now, consider the estimator of the squared angular power in each mode $\ell$ \cite{Thrane:2009fp,RomanoCornish},
\begin{align}
\label{eq:C_ell}
\hat{A}_{\fb,s,\ell}=\frac{1}{1+2\ell}\sum_{m=-\ell}^{\ell} \left[|\hat{a}_{\fb,s,\ell m}|^2 - (\Gamma^{-1}_{\fb,s})_{\ell m,\ell m}\right].
\end{align}
Denote with $\Ahat_{\fb,s}$ the $\ell_{\max}\times 1$-dimensional vector with components $\hat{A}_{\fb,s,\ell}$ for all $\ell= 1,\ldots, \ell_{\rm max}$. For each $s=1,\dots,S,$ the estimator $\Ahat_{\fb,s}$ has mean $\langle\Ahat_{\fb,s}\rangle$ and covariance matrix $\Sigma_{\fb,s}$, whose \((\ell,\ell')\)-th element is
\begin{eqnarray}
\label{eqn:cov_text}
\Sigma_{\fb,s,\ell\ell'} & = & \frac{1}{(1+2\ell)(1+2\ell')} 
\sum_{m ,m'} \Big( |( {\Gamma}_{\fb,s}
^{-1})_{\ell m,\ell' m'}|^2 \nonumber \\
& + & 2\Rel[a_{\fb,\ell m}^* \; 
({\Gamma}_{\fb,s}^{-1})_{\ell m,\ell' m'} \;
a_{\fb,\ell' m'}] \Big). 
\end{eqnarray}

 \textcolor{black}{One could always replace ${\Gamma}_{\fb,s}^{-1}$ in Eq.~\eqref{eq:C_ell} and Eq.~\eqref{eqn:cov_text} with its regularized counterpart if a regularization scheme was applied to invert the Fisher matrix \citep[Cf.][] {Thrane:2009fp,Agarwal_2023,PhysRevD.100.043541}.} The derivation of Eq.~\eqref{eqn:cov_text} is given in Appendix \ref{covariance_appendix}.
Previous works \citep[e.g.,][]{Agarwal_2023} have omitted the second term in the numerator of Eq.~\eqref{eqn:cov_text}, implicitly making the assumption that $a_{\fb,\ell m}$'s are complex Gaussian random variables with zero mean. However, this assumption is not consistent with estimating  
the coefficients $a_{\fb,\ell m}$ (as realized in our universe) via the maximum likelihood technique -- that is, likelihood maximization implicitly assumes the coefficients to be the (deterministic) limit to which the maximum likelihood estimates converge. In this manuscript, this inconsistency is avoided by treating $a_{\fb,\ell m}$ as fixed. 

\section{Statistical Methods}
\label{sec:III}

The main goal is to develop an inferential strategy for estimating angular power spectral models and propose a statistical test to assess whether the proposed model is consistent with the data derived from the estimator in Eq.~\eqref{eq:C_ell}.

\begin{equation}
\begin{aligned}
    \label{eqn:test}
    &H_0: \langle\Ahat_{\fb,s}\rangle=\bm{A}_\fb(\bmtheta) \quad \hbox{vs.}
    &H_a: \langle\Ahat_{\fb,s}\rangle\neq\bm{A}_\fb(\bmtheta)
\end{aligned}
\end{equation}
without assuming $\Ahat_{\fb,s}$ follows any specific distribution.

\subsection{\label{sec:III.1}Parameter and covariance estimation}

Begin by estimating the true value of the parameter vector $\bmtheta$, hereinafter denoted by $\bmthetao$, using the vectors $\Ahat_{\fb,s}$. For the moment, assume that the covariance matrices $\Sigma_{\fb,s}$ are known, and define the \emph{sphered errors} as
\begin{eqnarray}
\label{eqn:error}
   \bm{\epsilon}_{\fb,s}(\bm{\theta})=
   \Atilde_{\fb,s}-\Atilde_{\fb,s}(\bmtheta),
\end{eqnarray}
with 
\begin{equation}
\label{eqn:sphering}
\Atilde_{\fb,s}=\Sigma_{\fb,s}^{-1/2}\Ahat_{\fb,s}\quad\text{and} \quad \Atilde_{\fb,s}(\bmtheta)=\Sigma_{\fb,s}^{-1/2}\bm{A}_{\fb}(\bmtheta),\end{equation}
{\color{black}where $\Sigma_{\fb,s}^{-1/2}$ denotes the principal square-root matrix\footnote{\textcolor{black}{The computationally convenient formula is given by $\Sigma_{f_b,s}^\alpha=Q \Lambda^\alpha Q^T$, where $\Sigma_{f_b,s}=Q \Lambda Q^T$ is the Schur decomposition of a positive definite matrix, and $\Lambda^\alpha$ denotes the diagonal matrix whose entries are the $\alpha$th powers of the eigenvalues. 
Here, $\alpha=-1$ corresponds to the classical matrix inverse and $\alpha=-1/2$ corresponds to the inverse square root. }} of $\Sigma_{\fb,s}^{-1}$} and can be obtained via standard procedures\footnote{{\color{black}In the applications to follow, the square root matrix has been computed via the Schur method \citep[e.g.,][Ch. 6]{higham}. Nonetheless, other methods to construct the square root matrix, such as diagonalization, Jordan decomposition, etc., are also viable options.}}\citep{higham, Horn_Johnson}.
The transformation in Eq. \eqref{eqn:sphering} is known as \emph{whitening} or \emph{sphering} and enables the construction of a statistically orthonormalized version of the angular power spectrum $\Ahat_{\fb,s}$. In other words, under $H_0$ in Eq.~\eqref{eqn:test}, $\Atilde_{\fb,s}$ has mean given by $\Atilde_{\fb,s}(\bmtheta)$, its components are uncorrelated, and each has a unit variance.

The true value of the unknown parameter $\bmthetao$ can be estimated by solving the \emph{generalized nonlinear least squares} problem:
\begin{equation}
\begin{aligned}
    \label{optim}
    \thetahat
    &= \argmin_{\bmtheta}
\sum_{ b,s} \bm{\epsilon}_{\fb,s}(\bm{\theta})^T \bm{\epsilon}_{\fb,s}(\bm{\theta}).
    \end{aligned}
\end{equation} 
The optimization in Eq. \eqref{optim} is the same as a maximization of the likelihood 
of $\Ahat_{\fb,s}$ under the assumption of Gaussianity. Yet, its underlying rationale differs in that it aims to minimize the squared differences of the sphered errors, regardless of their distribution. 
In general, such a minimization does not enjoy a closed-form solution and is typically solved using iterative numerical methods \cite{1999numerical}.

The solution of Eq. \eqref{optim}, can be shown to be consistent\footnote{\textcolor{black}{
Intuitively, consistency means that as more data are collected, the estimator $\widehat{\bmtheta}$ eventually becomes arbitrarily close to the true value $\theta_0$. A simple sufficient condition for consistency is that $\widehat{\theta}$ is asymptotically unbiased and its variance converges to zero. }}
 under $H_0$. That is, if the model $\bm{A}_\fb(\bmtheta)$ is correctly specified, and if it is identifiable -- i.e., $\bm{A}_\fb(\bmtheta) =  \bm{A}_\fb(\bmtheta')$ implies $\bmtheta = \bmtheta'$ \citep[Cf.][Ch. 7]{davidson} -- then, $\thetahat$ converges in probability \cite[for a definition see][Ch. 2]{van2000asymptotic} to $\bmthetao$. Furthermore, under additional classical regularity conditions \citep[Cf.][Ch. 7]{davidson}, $\thetahat$ is approximately Gaussian distributed with mean $\bmthetao$ and its variance approaches zero as $N\rightarrow \infty$. This can be seen by taking the derivative of the \textcolor{black}{function being minimized in Eq. \eqref{optim}} 
with respect to $\bmtheta$, which shows that the generalized least squares estimator $\thetahat$ solves the system of $p$ estimating equations:
\begin{equation}
\label{est_eq}
   \frac{1}{\sqrt{N}}\sum_{b,s}\dot{\bm{A}}_{\fb,s}(\hat{\bmtheta})^T\bm{\epsilon}_{\fb,s}(\hat{\bmtheta}) =\bm{0},
\end{equation}
where $\dot{\bm{A}}_{\fb,s}(\hat{\bmtheta})$ denotes the $\ell_{\rm max} \times p$ matrix with columns corresponding to the partial derivatives of $\tilde{\bm{A}}_{\fb,s}(\bmtheta)$ evaluated at $\hat{\bmtheta}$, i.e., 

\begin{equation}
\dot{\bm{A}}_{\fb,s}(\hat{\bmtheta}) = \Sigma_{\fb,s}^{-1/2} \left[\frac{\partial}{\partial \theta_1}\bm{A}_\fb(\bmtheta) ,\ldots, \frac{\partial}{\partial \theta_p}\bm{A}_\fb(\bmtheta) \right]\biggl|_{\bmtheta=\hat{\bmtheta}},
\end{equation}
in which $\theta_j$, $j=1,\dots,p$ are the elements of $\bmtheta$. A first-order Taylor expansion of the left-hand side of Eq.~\eqref{est_eq} around $\thetahat=\bmthetao$ leads to:
    \begin{multline}
        \label{eqn:linearized}
    \frac{1}{\sqrt{N}} \sum_{b,s}\dot{\bm{A}}_{\fb,s}(\bmthetao)^T \errors_{\fb,s}(\bmthetao)-\bm{R}_p\sqrt{N}(\thetahat-\bmthetao) +o_p(1) 
    = \bm{0}
    \end{multline}
where the notation $o_p(1)$ is used to indicate random terms that converge to zero in probability as $N\rightarrow\infty$ and $\bm{R}_p$ is a $p\times p$ matrix given by
\begin{equation}
\bm{R}_p=\frac{1}{N}\sum_{b,s} \dot{\bm{A}}_{\fb,s}(\bmthetao)^T \dot{\bm{A}}_{\fb,s}(\bmthetao).
\end{equation} By rearranging the terms in Eq. \eqref{eqn:linearized}, we obtain the asymptotic representation:
\begin{equation}
\begin{aligned}
 \label{eqn:asymp}
    \sqrt{N}(\thetahat-\bmthetao)&=\frac{1}{\sqrt{N}}\bm{R}_p^{-1}\sum_{b,s}\dot{\bm{A}}_{\fb,s}(\bmthetao)^T\bm{\epsilon}_{\fb,s}(\bmthetao)+o_p(1).
\end{aligned}
\end{equation}
By the central limit theorem, it follows that $\thetahat$ is approximately Gaussian distributed with mean $\bmthetao$ and variance-covariance matrix $\bm{R}_p^{-1}/N$. 

When \(\Sigma^{-1}_{\fb,s}\) is unknown but a uniformly (for all $s$) consistent estimator is available, the statistical properties of \(\thetahat\) described above remain valid provided that this estimator is substituted for \(\Sigma^{-1}_{\fb,s}\) in Eq.~\eqref{eqn:sphering} \citep[Cf.][Ch.~7]{davidson}.

In our case, such an estimator for \( \Sigma^{-1}_{\fb,s}\) can be obtained by estimating $\Sigma_{\fb,s}$ from Eq.~\eqref{eqn:cov_text}, and a corresponding ‘plug-in’ estimator for $\Sigma_{\fb,s}$ is now briefly described. 
Notice that
\begin{equation}\label{equ:a-S}
    \hat{a}_{\fb,\ell m} 
= \sum_{s=1}^S \sum_{\ell', m'} \bigl(\Gamma_{\fb}^{-1}\bigr)_{\ell m, \ell' m'} \,\hat{X}_{\fb,s,\ell' m'},
\end{equation}
and hence $\hat{a}_{\fb,\ell m}$ depends on  the number of segments $S$.  In Appendix~\ref{app:d}, it is established that  $\hat{a}_{\fb,\ell m}$ converges in probability to $a_{\fb,\ell m}$ as $S \to \infty$\footnote{Note that as $S$ increases, so does the total observation time $T$.}.  The estimator has elements $(\hat{\Sigma}_{\fb,s} )_{\ell,\ell'}$ obtained by replacing $a_{\fb,\ell m}^*$ and $a_{\fb,\ell m}$ with $\hat{a}_{\fb,\ell m}^*$ and $\hat{a}_{\fb,\ell m}$ in Eq. \eqref{eqn:cov_text}. 
A theoretical analysis confirming the positive definiteness of these covariance estimators and a detailed proof that the inverse $\hat\Sigma^{-1}_{\fb,s}$ is a uniformly consistent estimator of $\Sigma^{-1}_{\fb,s}$ can be found in Appendix~\ref{app:d}.

\subsection{\label{sec:inference}Distribution-free testing of the angular power spectum}

In its essence, the statistical formulation of the problem provided in Sec.~\ref{sec:III.1} is that of a regression problem in which the estimated angular power spectrum, $\Ahat_{\fb,s}$, plays the role of the outcome whose mean we aim to describe through the model $\bm{A}_{\fb,s}(\bmtheta)$. We can, therefore, employ constructs from the theory of goodness-of-fit for regression to assess the validity of such a model. 

The main building blocks needed to devise a test for the hypotheses in Eq. \eqref{eqn:test} are the \emph{sphered residuals}. They can be obtained by replacing $\bmtheta$ in Eq. \eqref{eqn:error} with the solution of Eq. \eqref{optim}, i.e.,
\begin{eqnarray}
\label{residuals}
\residuals_{\fb,s}=\Atilde_{\fb,s}-\Atilde_{\fb,s}(\hat{\bmtheta}).
\end{eqnarray}
In principle, one could test the hypotheses in Eq. \eqref{eqn:test} using classical goodness-of-fit statistics constructed as the sum of some function of the sphered residuals -- with  Pearson's $\chi^2$ being possibly the most prominent example among those -- and derive their distribution under the hypothesis $H_0$ using the parametric bootstrap \citep[Cf.][]{freedman81} based on the generalized $\chi^2$ distribution derived in Appendix~\ref{generalized_chi}. Such statistics, however, are known to exhibit no statistical power\footnote{\textcolor{black}{In statistics, the concept of `power' refers to the sensitivity of a test in detecting departures from the model hypothesized under the null hypothesis.} } toward infinitely many possible `weak' deviations from the postulated model \citep[e.g.,][]{khm84,algeri2024pearson,acharya2024spectral}. That is, there are infinite possible ways in which the hypothesized model could be wrong, but such statistics would not be able to identify them. 
Test statistics based on the process of partial sums of the residuals \citep[e.g.,][]{loynes80,  stute97}, however, can restore statistical power \citep[e.g.,][]{khm84,algeri2024pearson,acharya2024spectral} and are thus recommended. 

\subsubsection{Model testing via partial sums of the residuals}
Let $\errors$ and $\residuals$ denote, respectively,  the $N\times 1$ stacked vectors of errors $\errors_{\fb,s}(\bm{\theta}_0)$ and residuals $\residuals_{\fb,s}$ for all frequency bins and data segments, that is, 
\begin{equation}
\errors = \left[\errors_{f_1,1}^T(\bm{\theta}_0),\dots,\errors_{f_B,1}^T(\bm{\theta}_0), \dots, \errors_{f_1,S}^T(\bm{\theta}_0),\dots, \errors_{f_B,S}^T(\bm{\theta}_0)\right]^T,
\end{equation}
and 
\begin{equation}
\residuals = \left[\residuals_{f_1,1}^T,\dots,\residuals_{f_B,1}^T, \dots, \residuals_{f_1,S}^T,\dots, \residuals_{f_B,S}^T\right]^T.
\end{equation}
The process\footnote{\color{black}The term ``process" here denotes a stochastic process, where $w_{k,N}$ is a sequence of discrete random variables indexed by $k$ and of length $N$.} of partial sums of the sphered residuals is 
\begin{equation}
\begin{split}
    \label{partial_sums}
    w_{k,N}&=  \frac{1}{\sqrt{N}}\residuals^T \mathbb{I}_{k,N},\\
  \text{with}\quad & \mathbb{I}_{k,N}=[\underbrace{1,\dots,1}_{k},\underbrace{0,\dots,0}_{N-k}]^T.
    \end{split}
\end{equation}
It consists of normalized cumulative sums of the first $k$ elements of $\residuals$ with $k=1,\dots,N$. 

Following the same argument of \citet{khm21}, it is possible to show that $w_{k,N}$ is asymptotically equal to a projection of the process of partial sums of the sphered errors given by
\begin{equation}
    \label{partial_sums_errors}
 \frac{1}{\sqrt{N}}\errors^T \mathbb{I}_{k,N}.
\end{equation}
As shown below, such a projection plays a fundamental role in deriving the limiting distribution of $w_{k,N}$.

Consider a first-order Taylor expansion of the residuals in Eq. \eqref{residuals} evaluated at $\thetahat=\bmthetao$, that is, 
\begin{equation}
\label{eqn:residuals_expansion}
\begin{split}
\residuals_{\fb,s}&=\errors_{\fb,s}(\bmthetao)-\dot{\bm{A}}_{\fb,s}(\bmthetao)(\thetahat-\bmthetao)+o_p(1).
\end{split}
\end{equation}
Substituting the right-hand side of Eq. \eqref{eqn:asymp} for $\thetahat-\bmthetao$ in Eq. \eqref{eqn:residuals_expansion} leads to
\begin{equation}
\label{eqn:projected_residuals}
\begin{split}
\residuals_{\fb,s}=&\errors_{\fb,s}(\bmthetao)\\
&-\frac{1}{N}\dot{\bm{A}}_{\fb,s}(\bmthetao)\Bigl[\bm{R}_p^{-1}\sum_{\fb,s} \dot{\bm{A}}_{\fb,s}(\bmthetao)^T\errors_{\fb,s}(\bmthetao)\Bigl]+o_p(1).\\
\end{split}
\end{equation}
Let $\dot{\bm{A}}(\bm{\theta}_0)$ be a $N\times p$ matrix denoting the
stacked matrix of $\dot{\bm{A}}_{\fb,s}(\bmthetao)$ across all frequency bins and time segments, i.e., 
\begin{equation}
    \dot{\bm{A}}(\bm{\theta}_0)
=
\begin{bmatrix}
\dot{\bm{A}}^T_{f_1,1}(\bm{\theta}_0),\cdots,\dot{\bm{A}}^T_{f_B,1}(\bm{\theta}_0),\cdots,\dot{\bm{A}}^T_{f_B,S}(\bm{\theta}_0)
\end{bmatrix}^T,
\end{equation}
then, the vector $\residuals$, with elements given by Eq.~\eqref{eqn:projected_residuals}, can be written as
\begin{equation}
\begin{split}
    \label{projected_residuals_vec}
\residuals&= \errors-\frac{1}{N}\dot{\bm{A}}(\bm{\theta}_0)\bm{R}_p^{-1}\dot{\bm{A}}(\bm{\theta}_0)^{T}\errors+o_p(1)\\
&=\errors-\sum_{j=1}^p\bm{\mu}_{\bmtheta_j} \bm{\mu}_{\bmtheta_j} ^{T}\errors+o_p(1),
\end{split}
\end{equation}
where $\bm{\mu}_{\bmtheta_j}$ corresponds to the $jth$ column of the matrix
\begin{equation}
\frac{1}{\sqrt{N}} \dot{\bm{A}}(\bm{\theta}_0)\bm{R}_p^{-1/2}.
\end{equation}

From Eq. \eqref{projected_residuals_vec}, it follows that the sphered residuals are asymptotically equal to a projection of the errors orthogonal to the vectors $\{\bm{\mu}_{\bmtheta_j}\}_{j=1}^p$. Similarly, the process $w_{k,N}$ can be written as:
\begin{equation}
    \label{projected}
    w_{k,N}=\frac{1}{\sqrt{N}}\errors^{T} \Bigl(\bm{I}_N-\sum_{j=1}^p\bm{\mu}_{\bmtheta_j} \bm{\mu}_{\bmtheta_j} ^{T}\Bigl) \mathbb{I}_{k,N}+o_p(1),
\end{equation}
where $\bm{I}_N$ denotes the $N\times N$ identity matrix. It follows that, in the limit, $w_{k,N}$ is a projection of the process in Eq. \eqref{partial_sums_errors} orthogonal to the vectors $\{\bm{\mu}_{\bmtheta_j}\}_{j=1}^p$. 

The process in Eq. \eqref{partial_sums_errors} can be shown to be asymptotically distributed as a Brownian motion if $H_0$ in Eq. \eqref{eqn:test} is true. This result follows from the fact that such a process consists of partial sums of weakly dependent\footnote{Intuitively, a sequence of random variables is said to be `weakly dependent' if the dependence among the elements of the sequence reduces and eventually ceases to exist for elements that are sufficiently far away from one another within the sequence.} quantities \citep[Cf.][]{philipp,bradley2005basic} with mean zero and finite variance.  The condition of weak dependence is satisfied since the subvectors $\errors_{\fb,s}(\bmthetao)$ of $\errors$ are all independent from one another. 
It follows that since  $w_{k,N}$ is a projection of the process in Eq. \eqref{partial_sums_errors}, its limiting distribution under $H_0$ is that of a projected Brownian motion.
Its mean and covariance are \begin{equation}
    \label{eqn:mean2}
    \begin{split}
    \langle w_{k,N} \rangle &\simeq 0 
\qquad \text{and}\\
 \langle w_{k,N} w_{k',N}\rangle
   & \simeq \frac{1}{N} \mathbb{I}_{k,N}^T\Bigl(\bm{I}_N-\sum_{j=1}^p \bm{\mu}_{\bmtheta_j} \bm{\mu}_{\bmtheta_j} ^{T}\Bigl)\mathbb{I}_{k',N} .
        \end{split}
\end{equation}
where the notation $\simeq$ indicates that the functions on both sides are asymptotically equal as $N \to \infty$. 


Test statistics to assess the validity of $H_0$ in Eq. \eqref{eqn:test} can be constructed by considering functionals of the process $w_{k,N}$ in Eq. \eqref{partial_sums}.
Two examples are 
\begin{equation}
    \begin{aligned}
K_N=\max\limits_k|w_{k,N}| \quad \text{and} \quad C_N=\frac{1}{N}\sum_{k=1}^N|w_{k,N}|^2,
\label{eqn:TS}
    \end{aligned}
\end{equation}
which can be seen as the counterparts of the Kolmogorov--Smirnov statistic and the Cramér-von Mises statistic in the context of regression. 

The second asymptotic equality in Eq. \eqref{eqn:mean2} implies that the distribution of $w_{k,N}$ under $H_0$  depends on the model being tested through the vectors $\{\bm{\mu}_{\bmtheta_j}\}_{j=1}^p$. 
Hence, the same is true for the distribution of $K_N$ and $C_N$. Such distributions, however, are unaffected by 
the distribution of the errors in $\errors$ and, therefore, they only depend on the mean and covariance of $\Ahat_{\fb,s}$. This implies that the statistics $K_N$ and $C_N$ always have the same distribution under $H_0$, regardless of the distribution of $\Ahat_{\fb,s}$ being approximately Gaussian in the limit.

To illustrate this aspect with an example, consider testing the models $\bm{A}^{M1}_{\fb}(\bmtheta)$ and $\bm{A}^{M2}_{\fb}(\bmtheta)$, with components given by 
\begin{equation}
\begin{aligned}
\label{eqn:simulated_true_mean}
\bm{A}^{M1}_{\fb,\ell}(\bmtheta) &=  (\theta_0+ \theta_1 \ell + \theta_2 \ell^2) \bar{\fb}^{2/3}, \\ 
\bm{A}^{M2}_{\fb,\ell}(\bmtheta)&=  \exp(\theta_0+ \theta_1 \ell + \theta_2 \ell^2)\bar{\fb}^{2/3}
\end{aligned}
\end{equation}
for $\ell=1,\ldots,8$, $\fb=[20,40],[40,60],\dots,[160,180]$,  $\bmtheta=(\theta_0,\theta_1,\theta_2)$ and $\bar{\fb}$ is the center of the frequency band $\fb$. 

Consider the sphered errors generated either from a standard multivariate Gaussian distribution or independently from Laplace distribution with zero mean  and unit variance. Let the covariance matrices $\bm{\Sigma}_{\fb,s}$ be generated from a Wishart distribution with 10 degrees of freedom, and these matrices are assumed to be known. Denote the unsphered errors as
\begin{equation}
\begin{aligned}
\label{eqn:simulated_true_err}
\bm{u}^{(1)}_{\fb,s} \sim N(\bm{0},\bm{\Sigma}_{\fb,s}), \quad \bm{u}^{(2)}_{\fb,s} \sim L(\bm{0},\bm{\Sigma}_{\fb,s}),
\end{aligned}
\end{equation}
for all $s=1,\dots,S$. We consider four datasets generated from the four possible additive combinations of models and error terms in Eqs. \eqref{eqn:simulated_true_mean}-\eqref{eqn:simulated_true_err} with $\bmtheta=(4,$ -$3,0.4)$. For each of these datasets, we perform $S=15$ repeated simulations, resulting in an effective sample size of $120,$ and a total of $N=960$. The datasets are generated 
$10^4$ times; each time, we calculate a single realization of the process $w_{k,N}$ in Eq.~\eqref{projected}, which yields one observation of the statistics $K_N$ and $C_N$ from Eq.~\eqref{eqn:TS}.

Figure \ref{Fig:1} shows the simulated distribution of the Kolmogorov–Smirnov and Cramér–von Mises statistics in Eq. \eqref{eqn:TS} for each of the four combinations of models and error terms specified in Eqs. \eqref{eqn:simulated_true_mean}-\eqref{eqn:simulated_true_err} under $H_0$ in \eqref{eqn:test} -- that is, the statistics are constructed so that the model being tested and the model used to simulate the data is the same. Observe that the distributions of the test statistics overlap when the underlying models are the same,  regardless of the error distributions. This observation aligns with Eqs. \eqref{eqn:mean2}-\eqref{eqn:TS}, which demonstrate that the distribution of the statistics depends solely on the model being tested through the vectors $\{\bm{\mu}_{\bm{\theta}j}\}_{j=1}^p$.


\begin{figure}
\includegraphics[scale=0.36]{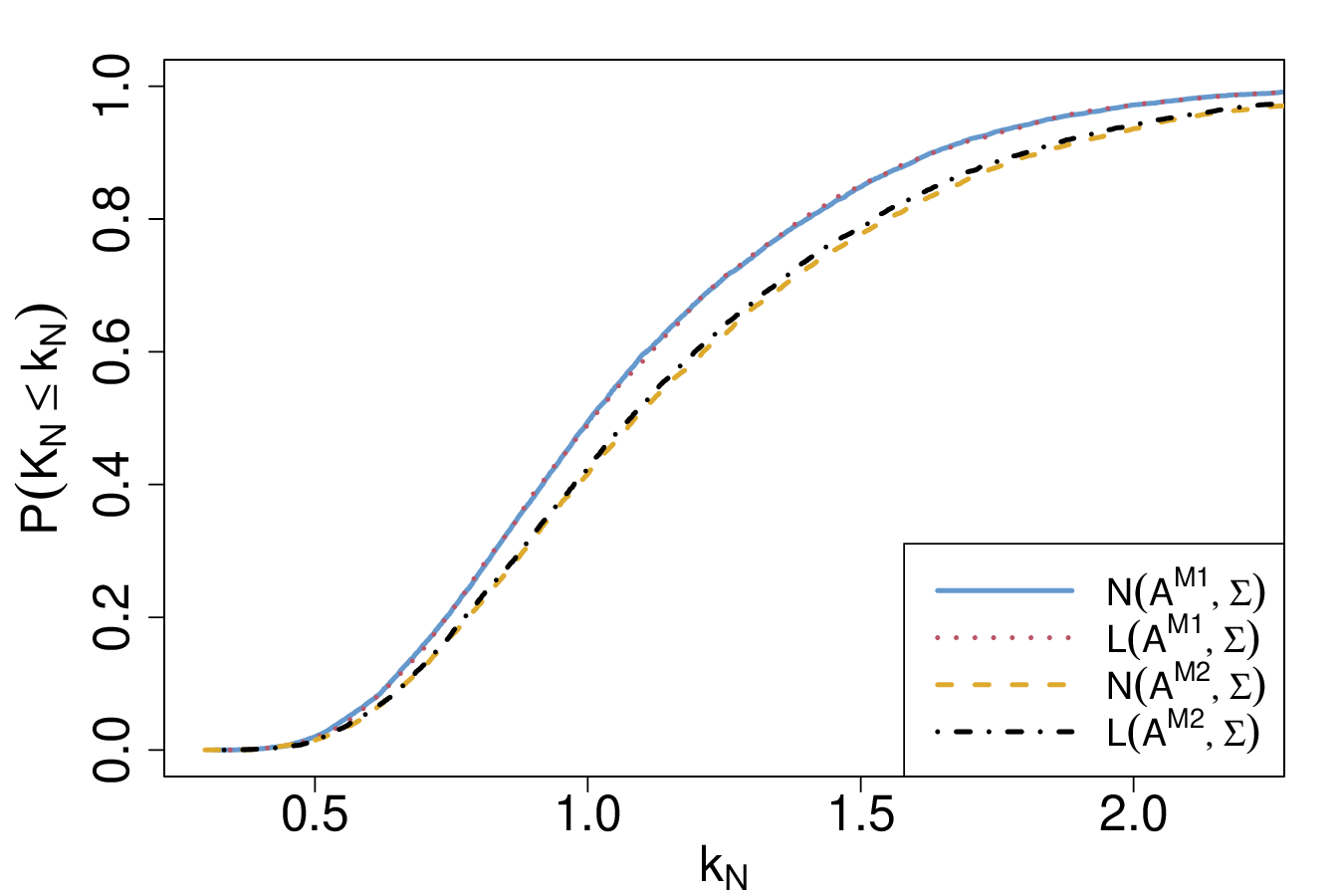}
\includegraphics[scale=0.36]{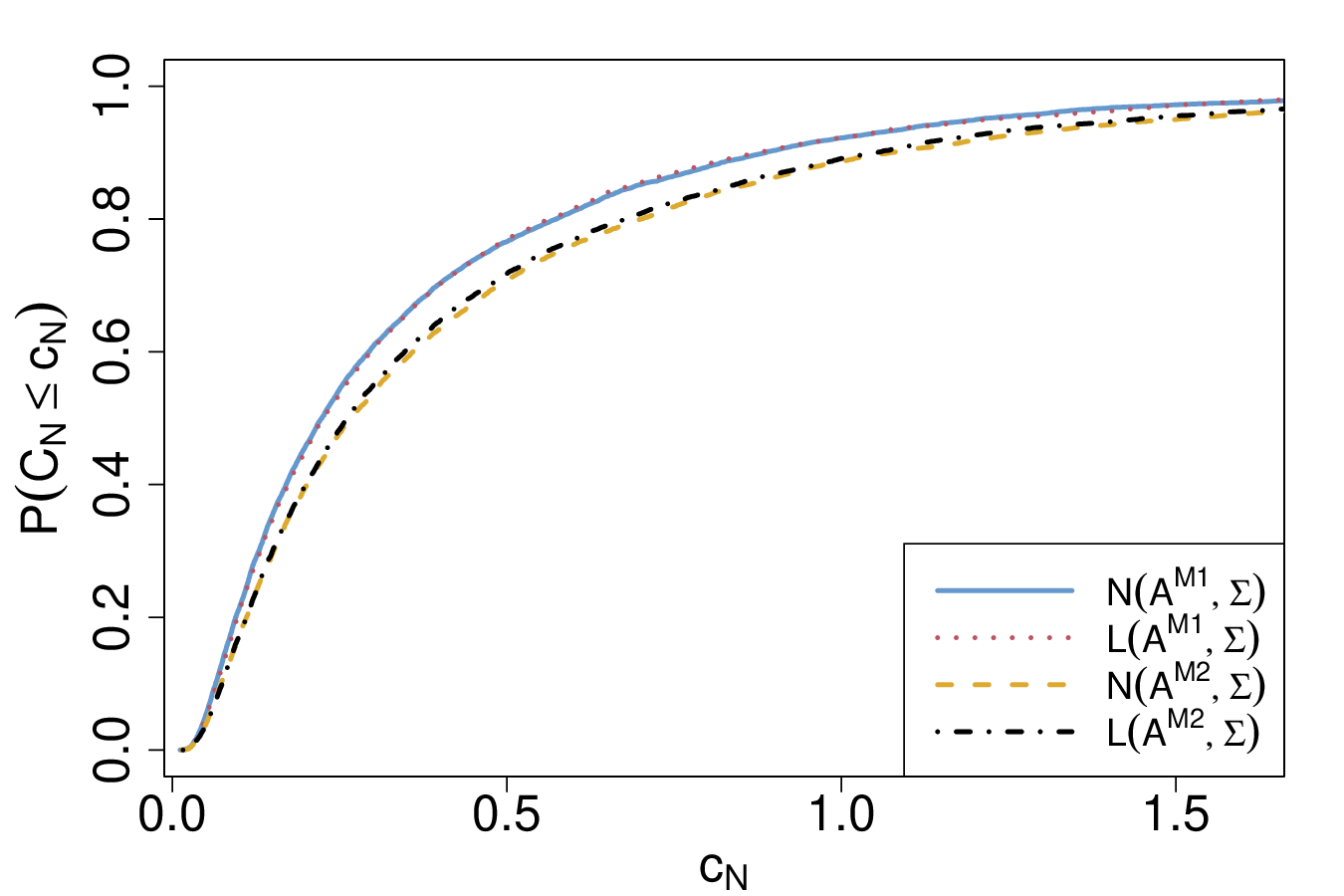}
\caption{Graphs of the simulated null cumulative distribution functions of the Kolmogorov--Smirnov (top) and Cramér--von Mises (bottom) statistics from the four combinations of $\bm{A}^{M1}_{\fb,s}$ and $\bm{u}^{(1)}_{\fb,s}$ (blue solid); $\bm{A}^{M1}_{\fb,s}$ and $\bm{u}^{(2)}_{\fb,s}$ (red dotted); $\bm{A}^{M2}_{\fb,s}$ and $\bm{u}^{(1)}_{\fb,s}$ (yellow dashed); $\bm{A}^{M2}_{\fb,s}$ and $\bm{u}^{(2)}_{\fb,s}$ (darkgreen dash-dotted).
}
\label{Fig:1}
\end{figure}

\subsubsection{\label{sec:III.2.2} Construction of distribution-free test statistics}

Since the null distribution of the test statistics in Eq. \eqref{eqn:TS} depends on the model being tested, it can only be derived on a case-by-case basis for each specific choice of the postulated angular power spectrum model $\bm{A}_\fb(\bmtheta)$. Unfortunately, this process can be computationally burdensome when testing several different and possibly complicated models. 
Nevertheless, it is possible to overcome this limitation by relying on the  \emph{Khmaladze-2 (K2) transform} \cite{khm21}. 

Given an arbitrary orthonormal set of vectors $\{\bm{r}_j\}_{j=1}^p$ in $\mathbb{R}^N$,  the K2 transformation allows us to map $w_{N,k}$
into a process of partial sums, hereafter denoted by $v_{k,N}$ with `standard' limiting distribution under the null hypothesis. Specifically, for sufficiently large $N$, its null distribution is well-approximated by that of a projected Brownian motion with covariance  
\begin{equation}
\label{eqn:cov_vn}
\begin{split}
  \langle v_{k,N} v_{k',N}\rangle \simeq \frac{1}{N}\mathbb{I}_{k,N} ^T \Bigl(\bm{I}_N-\sum_{j=1}^p \bm{r}_j\bm{r}_j^T\Bigl)\mathbb{I}_{k',N}.
\end{split}
\end{equation}
Thus, for any choice of  $\{\bm{r}_j\}_{j=1}^p$ independent from $\bm{A}_f(\theta)$ and such that $\sum_{j=1}^p \bm{r}_j\bm{r}_j^T \neq \bm{I}_N$,
  the distribution of $v_{k,N}$ is also independent from such a model.  One can then rely on the functionals of the newly constructed process $v_{k,N}$ 
  such as 
\begin{equation}
    \begin{aligned}
\widetilde{K}_N=\max\limits_k|v_{k,N}|, \quad \widetilde{C}_N=\frac{1}{N} \sum_{k=1}^N|v_{k,N}|^2, 
\label{eqn:TransTS}
    \end{aligned}
\end{equation}
to test the hypotheses in Eq. \eqref{eqn:test} in a distribution-free manner. As demonstrated in Section \ref{simulations}, this implies one can test different postulated models for the angular power spectrum by relying on a single simulation of the distribution of the statistics in Eq. \eqref{eqn:TransTS}.



The main steps involved in the construction of the process $v_{k,N}$, for any given orthonormal set $\{\bm{r}_j\}_{j=1}^p$, are outlined in \citet{PRL}. Below, we provide a thorough exposition of the technical aspects of such a transformation.

To ease the intuition, let us first consider the simple scenario where the unknown parameter $\bmtheta$ is one-dimensional. This situation is encountered, for example, when the interest is solely in estimating the amplitude of the astrophysical kernel used to generate theoretical values of the angular power spectrum (see Sec. \ref{sec:application}). 
Define the operator $U_{\bm{\mu}_{\bmtheta_1},\bm{r}_1}$ on $\mathbb{R}^N$ as
\begin{equation}
    \label{Umu1r1}
    U_{\bm{\mu}_{\bmtheta_1},\bm{r}_1} =\bm{I}-\frac{\langle\bm{\mu}_{\bmtheta_1}-\bm{r}_1, \cdot \rangle}{1-\langle\bm{\mu}_{\bmtheta_1},\bm{r}_1\rangle}(\bm{\mu}_{\bmtheta_1}-\bm{r}_1),
\end{equation}
where $\bm{I}$ is the identity operator. One can easily verify that  $U_{\bm{\mu}_{\bmtheta_1},\bm{r}_1}$
maps $\bm{\mu}_{\bmtheta_1}$ to $\bm{r}_1$, $\bm{r}_1$ to $\bm{\mu}_{\bmtheta_1}$, and leaves vectors orthogonal to both $\bm{\mu}_{\bmtheta_1}$ and $\bm{r}_1$ unchanged. 
As shown in Appendix \ref{app:e}, such an operator is unitary and self-adjoint. The \emph{K2-transformed errors} are given by 
\begin{equation}
\bm{e} = U_{\bm{\mu}_{\bmtheta_1},\bm{r}_1} \errors
\end{equation}
and we define the \emph{K2-transformed residuals} as 
\begin{equation}
\hat{\bm{e}} = U_{\bm{\mu}_{\bmtheta_1},\bm{r}_1} \residuals.
\end{equation}
Similarly to Eq.~\eqref{projected_residuals_vec}, the K2-transformed residuals are asymptotically equal to a projection of the K2-transformed errors. Specifically,
\begin{equation}
    \begin{aligned}
\hat{\bm{e}} &= U_{\bm{\mu}_{\bmtheta_1},\bm{r}_1} \errors -  U_{\bm{\mu}_{\bmtheta_1},\bm{r}_1} \bm{\mu}_{\bmtheta_1} \bm{\mu}_{\bmtheta_1} ^{T}\errors + o_p(1)\\
&=\bm{e}-\bm{r}_1\bm{r}_1^T\bm{e}  + o_p(1),
\label{eqn:trans_res_1d}
\end{aligned}
\end{equation}
where the second 
equality arises from the fact that 
\begin{equation}
    \begin{aligned}
U_{\bm{\mu}_{\bmtheta_1},\bm{r}_1} &\bm{\mu}_{\bmtheta_1} = \bm{r}_1 \\
\text{ and } \ \ \bm{\mu}_{\bmtheta_1} ^{T}\errors = \langle U_{\bm{\mu}_{\bmtheta_1},\bm{r}_1} &\bm{\mu}_{\bmtheta}, U_{\bm{\mu}_{\bmtheta_1},\bm{r}_1} \errors \rangle = \bm{r}_1^T\bm{e}.
\end{aligned}
\end{equation}
The process of partial sums of the K2-transformed residuals is 
\begin{align}
    \label{transformed_process_p}
   v_{k,N} &= \frac{1}{\sqrt{N}}\hat{\bm{e}}^{T}\mathbb{I}_{k,N} \\
   \notag
    &= \frac{1}{\sqrt{N}}\bm{e}^{T} \Bigl(I_N- \bm{r}_1\bm{r}_1 ^{T} \Bigl) \mathbb{I}_{k,N} +o_p(1);
\end{align}
its limiting distribution is that of a projected Brownian motion with zero mean and covariance approximately equal to the right-hand side of Eq. \eqref{eqn:cov_vn} for $p=1$.

Let us now consider the situation in which $\bm{\theta}$ is multidimensional -- as would be the case when estimating the amplitude, mean, and standard deviation of a Gaussian astrophysical kernel (see Sec. \ref{sec:application}, Eq. \eqref{eq:A_z_f}). 
Similarly to the one-dimensional case, the goal is to identify an operator, hereafter denoted by $\bm{U}^T_p$, such the resulting K2-transformed vector of residuals $\hat{\bm{e}}=\bm{U}^T_p \residuals$ is asymptotically equal to a projection of the vector of K2-transformed errors $\bm{e} =\bm{U}^T_p \errors$, i.e.,
\begin{equation}
\begin{aligned}
\label{eqn:K2errors}
\hat{\bm{e}} &= \bm{U}^T_p \errors - \sum_{j=1}^p \bm{U}^T_p \bm{\mu}_{\bmtheta_j} \bm{\mu}_{\bmtheta_j} ^{T}\errors + o_p(1)\\
 &=\bm{e}-\sum_{j=1}^p\bm{r}_j\bm{r}_j^T\bm{e}  + o_p(1).
\end{aligned}
\end{equation}
To serve this purpose, we aim to construct $\bm{U}^T_p$ so that it enables mapping the vectors $\{\bm{\mu}_{\bmtheta_j}\}_{j=1}^p$ to the arbitrarily chosen orthonormal set $\{\bm{r}_j\}_{j=1}^p$.
A -- somewhat naive -- first attempt, could be that of combining operators of the form $\bm{U}_{\bm{\mu}_{\bmtheta_{j}},\bm{r}_{j}}$, constructed as in Eq. \eqref{Umu1r1}. Such an approach, however,  does not yield the desired result. To see this, consider 
\begin{equation}
\bm{U}_{\bm{\mu}_{\bmtheta_{1}},\bm{r}_{1}}\bm{U}_{\bm{\mu}_{\bmtheta_{2}},\bm{r}_{2}}\bm{\mu}_{\bmtheta_{2}}=\bm{U}_{\bm{\mu}_{\bmtheta_{1}},\bm{r}_{1}}\bm{r}_{2}\neq \bm{r}_{2}.
\end{equation}
To overcome this difficulty, consider a set of vectors $\{\widetilde{\bm{r}}_j\}_{j=1}^p$ constructed so that each $\widetilde{\bm{r}}_j$ is orthogonal to $\bm{\mu}_{\bmtheta_k}$, for all $k<j$.
Specifically, we choose  $\widetilde{\bm{r}}_{1} =  \bm{r}_{1}$
and 
\begin{equation}
    \begin{aligned}
    \label{eqn:tilde_rj}
\widetilde{\bm{r}}_{j} =\bm{U}_{j-1} \bm{r}_{j}
\end{aligned}
\end{equation}
for $j=2, \ldots, p$,   with $\bm{U}_j$ denoting the product operator:
\begin{equation}
    \begin{aligned}
    \label{eqn:Uj}\bm{U}_j=\bm{U}_{\bm{\mu}_{\bmtheta_{j}}, \widetilde{\bm{r}}_{j}}  \ldots  \bm{U}_{\bm{\mu}_{\bmtheta_1}, \widetilde{\bm{r}}_{1}}\end{aligned}\end{equation}
with each operator $\bm{U}_{\bm{\mu}_{\bmtheta_{j}}, \widetilde{\bm{r}}_{j}}$ acting on everything on its right.
Since the operator $\bm{U}_j$ involves a product of unitary operators, it is also unitary. Its adjoint operator is given by its transpose:
\begin{equation}
    \begin{aligned}
    \label{eqn:Ujstar}
\bm{U}^T_j = \bm{U}_{\bm{\mu}_{\bmtheta_1},\widetilde{\bm{r}}_{1}}\ldots  \bm{U}_{\bm{\mu}_{\bmtheta_{j}}, \widetilde{\bm{r}}_{j}}.
\end{aligned}
\end{equation}

One can verify that the vectors $(\widetilde{\bm{r}}_j)_{j=1}^p$ are orthogonal to $\bm{\mu}_{\bmtheta_k}$, for all $k<j$  by noticing that
\begin{equation}
    \begin{aligned}
    \langle \widetilde{\bm{r}}_{2}, \bm{\mu}_{\bmtheta_1} \rangle &= \langle \bm{r}_2, U_{\bm{\mu}_{\bmtheta_1}, \widetilde{\bm{r}}_1} \bm{\mu}_{\bmtheta_1}\rangle = \langle \bm{r}_2, \bm{r}_1\rangle = 0,\\
     \langle \widetilde{\bm{r}}_{3}, \bm{\mu}_{\bmtheta_1} \rangle &= \langle U_{\bm{\mu}_{\bmtheta_1}, \widetilde{\bm{r}}_1} \bm{r}_{3}, U_{\bm{\mu}_{\bmtheta_2}, \widetilde{\bm{r}}_2} \bm{\mu}_{\bmtheta_1} \rangle= \langle U_{\bm{\mu}_{\bmtheta_1}, \widetilde{\bm{r}}_1} \bm{r}_{3},\bm{\mu}_{\bmtheta_1}\rangle\\
     &= \langle  \bm{r}_{3},U_{\bm{\mu}_{\bmtheta_1}, \widetilde{\bm{r}}_1}\bm{\mu}_{\bmtheta_1}\rangle
     =\langle  \bm{r}_{3}, r_1\rangle=0,\\
    \langle \widetilde{\bm{r}}_{3}, \bm{\mu}_{\bmtheta_2} \rangle &= \langle U_{\bm{\mu}_{\bmtheta_1}, \widetilde{\bm{r}}_1} \bm{r}_{3}, U_{\bm{\mu}_{\bmtheta_2}, \widetilde{\bm{r}}_2} \bm{\mu}_{\bmtheta_2} \rangle= \langle U_{\bm{\mu}_{\bmtheta_1}, \widetilde{\bm{r}}_1} \bm{r}_{3}, \widetilde{\bm{r}}_2 \rangle\\
    &= \langle  \bm{r}_{3},   \bm{r}_{2}\rangle=0,
\end{aligned}
\end{equation}
and proceeding by induction.

The operator $\bm{U}_p$ defined in Eq. \eqref{eqn:Uj}, for $j=p$,
 maps vectors $\bm{r}_j$ to  $\bm{\mu}_{\bmtheta_j} $ for all $j=1,\ldots,p,$. Specifically,
\begin{equation}
\begin{aligned}
    \bm{U}_p \bm{r}_j &= U_{\bm{\mu}_{\bmtheta_p}, \bm{\widetilde{r}}_p}  \cdots  U_{\bm{\mu}_{\bmtheta_1}, \widetilde{\bm{r}}_1}\bm{r}_j \\
    &= U_{\bm{\mu}_{\bmtheta_p}, \bm{\widetilde{r}}_p}  \cdots  U_{\bm{\mu}_{\bmtheta_j}, \widetilde{\bm{r}}_j}\widetilde{\bm{r}}_j \\ 
    &= U_{\bm{\mu}_{\bmtheta_p}, \bm{\widetilde{r}}_p}  \cdots  U_{\bm{\mu}_{\bmtheta_{j+1}}, \widetilde{\bm{r}}_{j+1}} \bm{\mu}_{\bmtheta_j} = \bm{\mu}_{\bmtheta_j},
\end{aligned}
\end{equation}
where the second equality follows from the definition of $\widetilde{\bm{r}}_{j}$ in Eq. \eqref{eqn:tilde_rj} and the third equality follows from the fact that each operator $U_{\bm{\mu}_{\bmtheta_j},\widetilde{\bm{r}}_{j}}$ maps $\widetilde{\bm{r}}_{j}$ to $\bm{\mu}_{\bmtheta_j}$. The fourth equality holds since $\bm{\mu}_{\bmtheta_j}$ is orthogonal to $\bm{\mu}_{\bmtheta_{j'}}$, with $j'\neq j$ and to all $\widetilde{\bm{r}}_{j'}$ with $j'> j$, and the operators $U_{\bm{\mu}_{\bmtheta_j},\widetilde{\bm{r}}_{j}}$ keep the vectors that are orthogonal to both $\bm{\mu}_{\bmtheta_j}$ and $\widetilde{\bm{r}}_{j}$ unchanged. 
Moreover, since
$\bm{\mu}_{\bmtheta_j} = \bm{U}_p \bm{r}_j $ then
\begin{equation}
    \bm{U}_p^{T} \bm{\mu}_{\bmtheta_j}  = \bm{U}_p^{T} \bm{U}_p \bm{r}_j = \bm{r}_j; 
\end{equation}
hence, $\bm{U}_p^{T}$ maps the vectors $\bm{\mu}_{\bmtheta_j} $ to $\bm{r}_j$ for all $j=1,\ldots,p$, thereby providing the desired transformation of the residuals  in Eq. \eqref{eqn:K2errors} leading to the process  
\begin{align}
    \label{transformed_process_pp}
   v_{k,N} &= \frac{1}{\sqrt{N}}\hat{\bm{e}}^{T}\mathbb{I}_{k,N} \\
   \notag
   &= \frac{1}{\sqrt{N}}\bm{e}^{T} \Bigl(\bm{I}_N-\sum_{j=1}^p\bm{r}_{j} \bm{r}_{j} ^{T}\Bigl) \mathbb{I}_{k,N} +o_p(1).
\end{align}
Such a process generalizes Eq. \eqref{transformed_process_p} to the multidimensional setting. Under $H_0$, the process $v_{k,N}$ in Eq. \eqref{transformed_process_pp} converges to a projected Brownian motion with covariance approximately equal to the right-hand side of Eq. \eqref{eqn:cov_vn}.



The limiting null distributions of the test statistics $\widetilde{K}_N$ and $\widetilde{C}_N$ in Eq. \eqref{eqn:TransTS} can be easily derived numerically. Specifically, let $\{\boldsymbol{e}^{(b)}\}_{b=1}^B$ be a set of $B$ vectors of errors simulated from a standard multivariate normal. Realizations of the limiting process of $v_{k,N}$  can be obtained by calculating
\begin{equation}
\label{eqn:limit_process_transformed}
   v_{k,N}^{(b)}= \frac{1}{\sqrt{N}}\bigl[\bm{e}^{(b)}-\sum_{j=1}^p\bm{r}_j\bm{r}_j^T\bm{e}^{(b)}\bigl]^{T}\mathbb{I}_{k,N},
\end{equation}
for $b=1,\dots,B$. Each $v_{k,N}^{(b)}$ is 
a realization of a mean-zero Gaussian process with covariance as in Eq. \eqref{eqn:cov_vn}. Thus, for sufficiently large $N$, it converges to the same projected  Brownian motion as $v_{k,N}$. By replacing $v_{k,N}$ with $v_{k,N}^{(b)}$ in Eq. \eqref{eqn:TransTS}, we can obtain $B$ replicates of the test statistics, denoted as $\widetilde{K}^{(b)}_{N}$, respectively. Therefore, letting $\widetilde{K}^{\text{obs}}_{N}$ and $\widetilde{C}^{\text{obs}}_{N}$
be their values obtained on the observed data, the corresponding p-values are:
\begin{equation}
 \frac{1+\sum_{b=1}^B \mathbbm{1}_{\bigl\{\widetilde{K}_{N}^{(b)} \geq \widetilde{K}^{\text{obs}}_{N}\bigl\}}}{1+B}\quad\text{and}\quad \frac{1+\sum_{b=1}^B \mathbbm{1}_{\bigl\{\widetilde{C}_{N}^{(b)} \geq \widetilde{C}^{\text{obs}}_{N}\bigl\}}}{1+B}
 \label{eqn:pval}
\end{equation}
with $\mathbbm{1}$ denoting the indicator function.

\subsection{\label{simulations} Simulation Studies}
\subsubsection{Distribution-free Property}

To demonstrate the validity of the distribution-free property of the method outlined in Section~\ref{sec:III.2.2}, let us revisit the example introduced in Section~\ref{sec:inference}. Under the same setup, we simulate the null distributions of the statistics $\widetilde{K}_N$ and $\widetilde{C}_N$ in Eq.~\eqref{eqn:TransTS} 
with $v_{k,N}$ constructed as in Eq. \eqref{transformed_process_pp} for each of the four combinations of models and error structures
given in Eqs.~\eqref{eqn:simulated_true_mean}-\eqref{eqn:simulated_true_err}. {\color{black}We adopt the orthonormal basis $\{\bm{r}_j\}_{j=1}^p$ constructed following \citet{khm21}, where 
$\boldsymbol{r}_1 = (1/\sqrt{N},\dots,1/\sqrt{N})^T$ 
and $\boldsymbol{r}_2$ are the vector of elements
\begin{equation}
\label{eqn:orthbasis}
\bm{r}_{2,n}=\sqrt{\frac{12N}{N^2-1}}\left(\frac{n}{N}-\frac{N+1}{2 N}\right), \ \ n=1,\ldots, N.
\end{equation} 
The remaining vectors $\bm{r}_3, \ldots, \bm{r}_p$ are obtained by applying the Gram–Schmidt procedure to successive powers of $\bm{r}_2$. This orthonormal basis are also adopted for the following simulation studies and the data analysis.}

For the sake of comparison, we also simulate the null distribution of $\widetilde{K}_N$ and $\widetilde{C}_N$ using the Monte Carlo procedure described at the end of Section~\ref{sec:III.2.2} and for which the errors are generated from a standard Gaussian. Figure~\ref{Fig:2} shows that regardless of the error structure or model considered, the graphs of the simulated null distributions overlap for both the Kolmogorov-Smirnov and the Cramér-Von Mises statistics. {\color{black}In contrast to Figure~\ref{Fig:1}, where the curves for M1 and M2 differ because the classical processes $w_{k,N}$ converge to Brownian motions with a common mean but model-dependent variance, the transformed processes $v_{k,N}$ share the same mean and variance across models, leading to overlapping null distribution of the test statistics considered.} This confirms that the proposed tests are, for sufficiently large $N$, distribution-free, and the proposed Monte Carlo scheme accurately approximates the limiting null distributions of the test statistics based on $v_{k,N}$. The  time required by the latter amounts to 3.74 seconds. 
{\color{black}All computations were performed on a MacBook Pro (16-inch, 2019) equipped with a 2.3 GHz 8-core Intel Core i9 processor, 16 GB of 2667 MHz DDR4 memory, and an Intel UHD Graphics 630 (1536 MB) card.}


\begin{figure}
\includegraphics[scale=0.38]{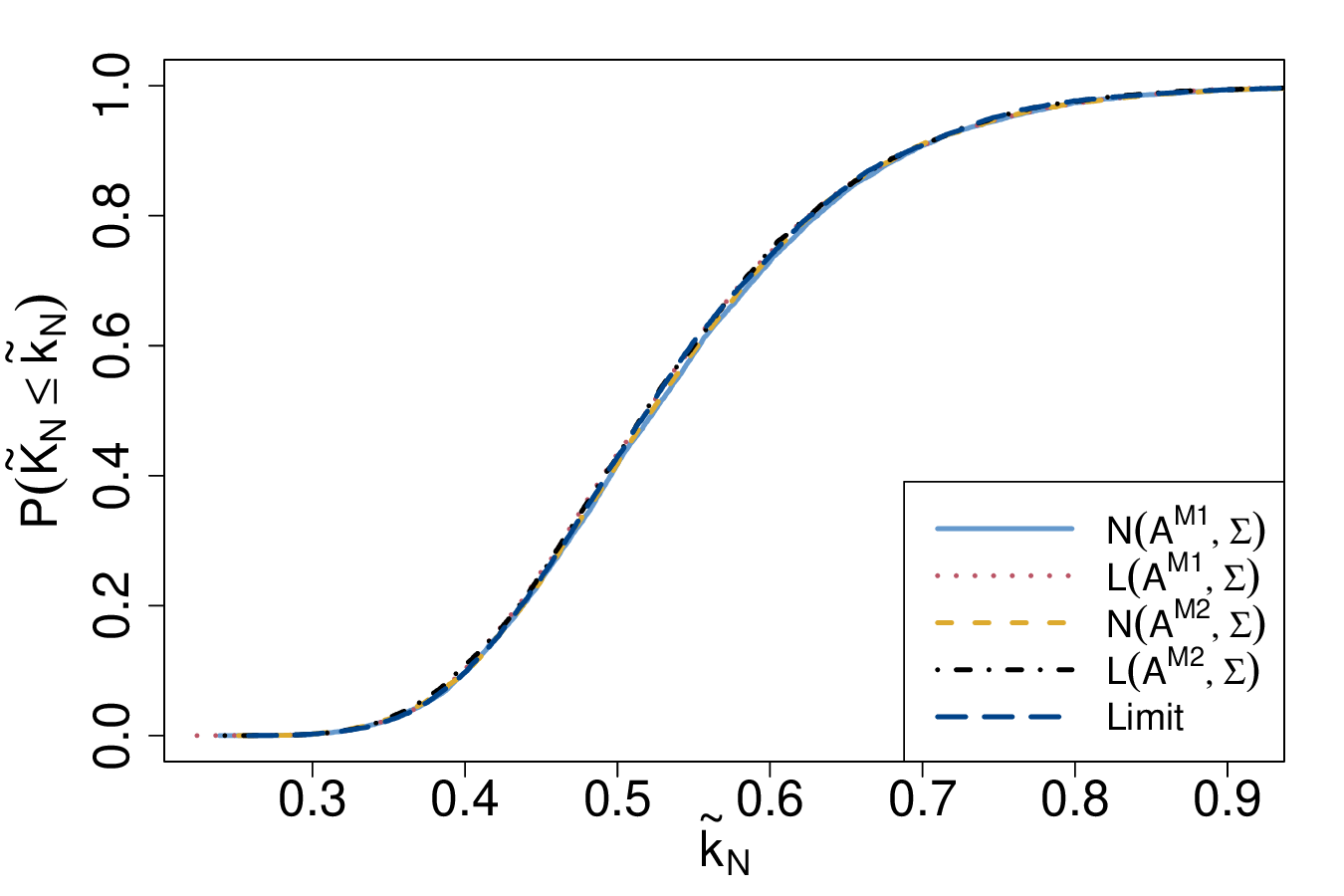}
\includegraphics[scale=0.38]{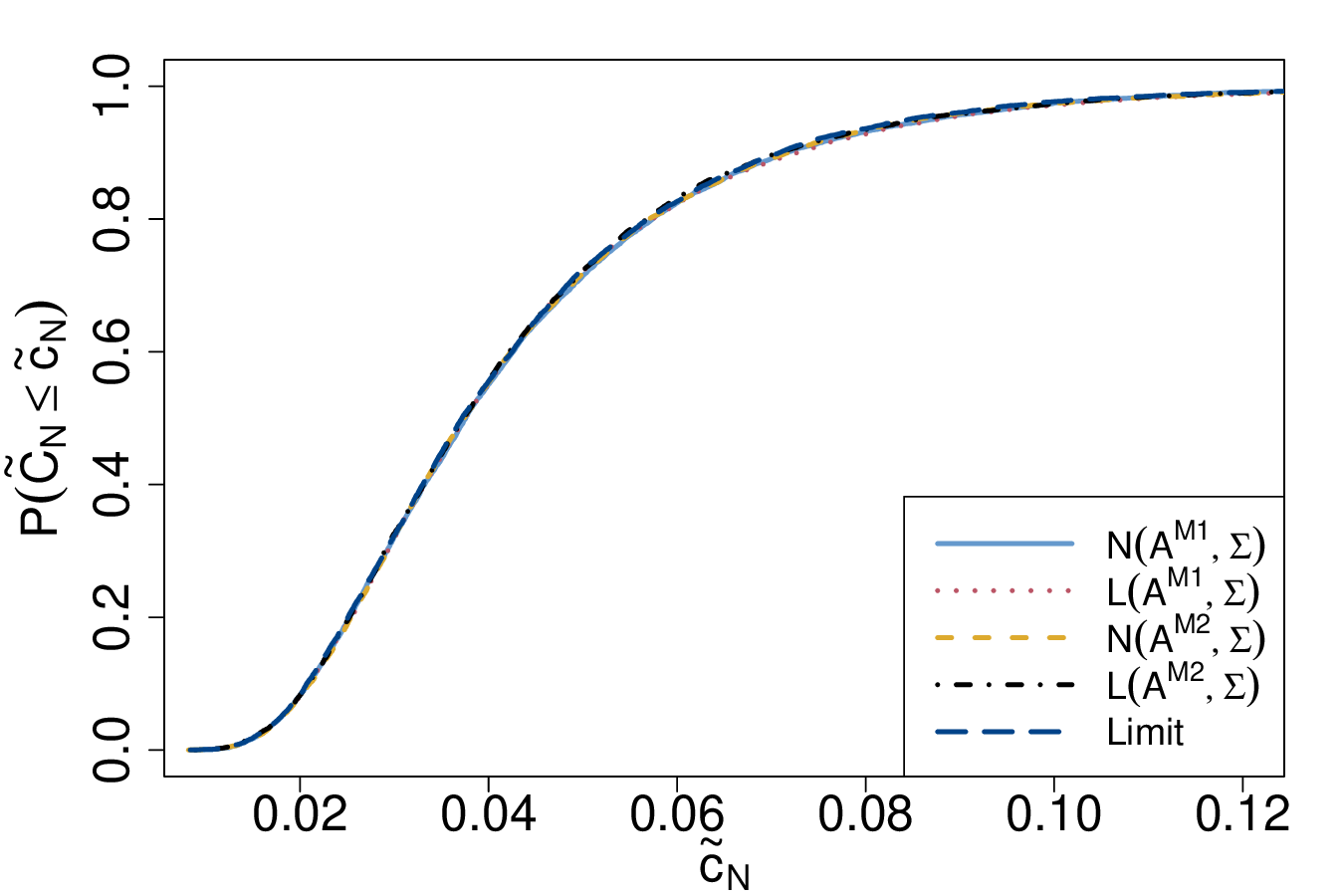}
\caption{Graphs of the simulated null cumulative distribution functions of the Kolmogorov-Smirnov (top) and the Cramér-Von Mises (bottom) statistics in Eq.~\eqref{eqn:TransTS} for the four combinations of models and error structures in Eqs.~\eqref{eqn:simulated_true_mean}-\eqref{eqn:simulated_true_err}, alongside their limiting null distribution. 
}
\label{Fig:2}
\end{figure}

\subsubsection{Statistical properties \label{sec:Power}}

The simulation studies that follow aim to assess the statistical properties of tests based on the statistics $\widetilde{K}_N$ and $\widetilde{C}_N$ in Eq.~\eqref{eqn:TransTS}.

To ensure that the probability of a false rejection\footnote{\textcolor{black}{A false rejection means the statistical test concludes that the hypothesized model does not fit the data well  even though the model is correct.}} -- known in statistics as the probability of Type I error -- does not exceed the pre-determined significance level, we generate  $10^4$ datasets from the additive combination of the mean function $\bm{A}^{M1}_{\fb,\ell}(\bmtheta)$ and the Gaussian error $\bm{u}^{(1)}_{\fb,s} $ in Eqs.~\eqref{eqn:simulated_true_mean}-\eqref{eqn:simulated_true_err}. We test the validity of $\bm{A}^{M1}_{\fb,\ell}(\bmtheta)$ by means of the statistics $\widetilde{K}_N$ and $\widetilde{C}_N$. Table~\ref{tab:table1} shows the simulated probabilities of rejecting $\bm{A}^{M1}_{\fb,\ell}(\bmtheta)$ when choosing different significance levels. The results indicate that the simulated Type I error probabilities closely match the nominal levels.

\begin{table}[t!]
\caption{Probability of Type I error for $\widetilde{K}_N$ and $\widetilde{C}_N$ at various significant levels}
\label{tab:table1}
\begin{ruledtabular}
\begin{tabular}{ccccc}
 $\alpha$ & $0.001$ & $0.01$ & $0.05$ & $0.1$ \\
\hline
\rule{0pt}{1.1em} $\widetilde{K}_N$ &  0.0009 & 0.008 & 0.046 & 0.093 \\
$\widetilde{C}_N $& 0.0007 & 0.009 & 0.047 & 0.092\\
\end{tabular}
\end{ruledtabular}
\end{table}

Next, we investigate the statistical power of the proposed statistics -- that is, the probability of correctly rejecting the proposed model when misspecified. In our simulation, we test 
$\bm{A}^{M1}_{\fb,\ell}(\bmtheta)$ while generating data using three different models for the mean. Specifically, we consider
\begin{equation}
\begin{aligned}
    \label{eqn:power_mean}
G_1: {\langle\Ahat_{f}\rangle}_\ell & =  (4 - 3 \ell + 0.4 \ell^2 + 0.02 \ell^3) \bar{\fb}^{(2/3)}; \\
G_2: {\langle\Ahat_{f}\rangle}_\ell & =  \left(4 -3 \exp(0.23\ell)\right) \bar{\fb}^{(2/3)}; \\
G_3: {\langle\Ahat_{f}\rangle}_\ell & =  (4-3 \ell^{2.07}) \bar{\fb}^{(2/3)}.
\end{aligned}
\end{equation}
A Monte Carlo simulation involving $10^4$ replicates is performed for each of the models in Eq. \eqref{eqn:power_mean} and combined with the Gaussian error structure in Eq. \eqref{eqn:simulated_true_err}. The simulated power of the Kolmogorov--Smirnov and Cramér--von Mises test statistics are reported in Tables~\ref{tab:table3_ks_power} and \ref{tab:table3_cvm_power}, respectively, for different significance levels $\alpha$. For both test statistics, the simulated statistical power is close to one when the significance level is 0.05 and \textcolor{black}{can be as high as $0.776$ even at smaller significance levels $\alpha=0.001$}.

\begin{table}[t!]
\caption{Statistical power of $\widetilde{K}_N$ when testing $\bm{A}^{M1}_{\fb,\ell}(\bmtheta)$ and generating data from $G_1$, $G_2,$ $G_3$ at various significant levels}
\label{tab:table3_ks_power}
\begin{ruledtabular}
\begin{tabular}{cccccc}
 $\alpha$ & $0.001$ & $0.01$ & $0.05$ & $0.1$ & \\
\hline
$G_1$  & 0.639 & 0.957 & 1.000 & 1.000 &\\
$G_2$ & 0.407 & 0.817 & 0.981 & 0.998 &\\
$G_3$ & 0.776 & 0.989 & 1.000 & 1.000 &\\
\end{tabular}
\end{ruledtabular}
\end{table}

\begin{table}[t!]
\caption{Statistical power of $\widetilde{C}_N$ when testing $\bm{A}^{M1}_{\fb,\ell}(\bmtheta)$ and generating data from $G_1$, $G_2,$ $G_3$ at various significant levels}
\label{tab:table3_cvm_power}
\begin{ruledtabular}
\begin{tabular}{cccccc}
 $\alpha$ & $0.001$ & $0.01$ & $0.05$ & $0.1$ & \\
\hline
$G_1$  & 0.232 & 0.588 & 0.895 & 0.971 &\\
$G_2$ & 0.105 & 0.364 & 0.730 & 0.881&\\
$G_3$ & 0.397    & 0.753 & 0.957 & 0.991& \\
\end{tabular}
\end{ruledtabular}
\end{table}

\section{\label{sec:application} Testing spectral models for the SGWB}
We now apply the methodology described in Sec. \ref{sec:III} to the case of a realistic astrophysical SGWB model and the data obtained by the Advanced LIGO and Advanced Virgo GW detectors in their third observing run (O3). \textcolor{black}{O3 includes data from both Advanced LIGO detectors as well as the Advanced Virgo detector, collected from April 1, 2019, to March 27, 2020, with a month-long pause in October 2019. Aside from this pause, the duty cycles of the three detectors were 77\%, 75\%, and 76\% for LIGO Livingston, LIGO Hanford, and Virgo, respectively \cite{Davis_2021}. For this study, we subdivide the O3 data set into $S=$15 segments, each corresponding to approximately 10 days of detector live time. This choice was made to ensure that each part of the sidereal day included data in each of the 15 segments. We also divide the frequency range 20-180 Hz into $B=8$ discrete frequency bins, each 20 Hz wide. We assume the same time-domain and frequency-domain data quality cuts that were applied during other O3 anisotropic SGWB search analyses \cite{O3stochdir}.}  

Let us turn our attention to the theoretical model we wish to test. The energy density of the stochastic gravitational-wave background is defined as

\begin{align}
	    \Omega_{\text{GW}}(\hat\Omega,f)=\frac{f}{\rho_c}\frac{d^3 \rho_{\text{GW}}}{d^2 \hat\Omega df}(\hat\Omega,f)=\frac{\bar{\Omega}_{\text{GW}}(f)}{4\pi}+\delta \Omega_{\text{GW}}(\hat\Omega,f),
\end{align}
where $\rho_{\text{GW}}$ is the GW energy density, $f$ is a continuous variable for GW frequency, and $\rho_c$ is the critical density of the universe today. We use $f$ to denote frequency as a continuous variable in contrast to $\fb$, which represents a discrete frequency band with width 20Hz and midpoint $\bar{\fb}$. The energy density is split into an isotropic component $\bar{\Omega}_{\text{GW}}$ and anisotropic perturbations $\delta \Omega_{\text{GW}}$ \cite{Cusin:2017mjm, Cusin:2019jpv}:
\begin{equation}
    \begin{aligned}
	    &\bar{\Omega}_{\text{GW}}(f) = \int \mathrm{d}r \, \partial_r \bar{\Omega}_{\text{GW}}(f,r),\\
	    &\partial_r \bar{\Omega}_{\text{GW}}(f,r)= \frac{f}{\rho_c}\,\mathcal{B}(f,r;\bmtheta)\,,
     \label{BackandPert}
\end{aligned}
\end{equation}
where $r$ denotes the conformal distance
and the function $\mathcal{B}(f,r;\bmtheta)$ is an astrophysical kernel with free parameters $\bmtheta$ that contains information on the local production of GWs at galaxy scales (we set the speed of light $c=1$). Following~\cite{Cusin:2017fwz}, the angular power spectrum is then given by
\begin{eqnarray}
\label{eq:modelCl}
 A_{\ell}(f; \bmtheta) & = & \frac{2}{\pi} \int dk k^2 \, |\delta \Omega_{\text{GW}\,,{\ell}}(k,f; \bmtheta)|^2 ,\hbox{ and}\\
\delta \Omega_{\text{GW}\,,\ell}(k,f; \theta) & = & \frac{f}{4\pi \rho_c} \int dr \, \mathcal{B}(f,r; \bmtheta) \, \big[b(r)\,\delta_{m,k}(r)j_{\ell}(k r) \big], \nonumber \\
\end{eqnarray}
where $k$ is the wavenumber, $j_{\ell}$ are spherical Bessel functions, and $\delta_{m,k}$ is the dark-matter over-density, related to galaxy overdensity via the bias factor that we assume to be scale-independent and with redshift evolution given by $b(z)=b_0 \sqrt{1+z}$, where $b_0=1.5$ if we replace the conformal distance $r$ with redshift $z$ \cite{yang2023}. Note that  Eq.~\eqref{eq:modelCl} is the model to be tested. 
The astrophysical kernel can be computed from the galaxy distribution and galaxy GW luminosity; for our purposes, we also replace the conformal distance $r$ with redshift $z$ and adopt the empirical parameterization \cite{Cusin:2017mjm, Cusin:2019jpv, yang2023}
\begin{align}\label{eq:A_z_f}
\mathcal{B}(f, z; \bmtheta) = 
A_{\text{\rm max}} \, f^{-1/3} \, e^{-(z-z_c)^2/2\sigma_z^2}\,,
\end{align}
with the free parameters $\theta = (A_{\text{\rm max}}, z_c,\sigma_z)$.

A direct detection of the SGWB has not yet been achieved, so the data from O3 are consistent with noise \cite{O3stochdir}. Therefore, we are unable to use O3 data with the method developed in Sec. \ref{sec:III} to validate the theoretical model for the angular power spectrum. Instead, we add a simulated signal (generated from the theoretical model itself) to the noise-only O3 data, and then apply our distribution-free methodology to these simulated data. The addition of signal to noise must be done in terms of maps of GW power rather than the angular power spectrum, since the latter physically represents \textit{squared} power rather than simply power. We may then represent the sum of noise and the injected signal as
	\begin{align}\label{eq:alm_inj}
		\hat{a}_{\fb,s, \ell m}^{\rm inj}=\hat{a}_{\fb,s, \ell m} + {a}_{\fb, \ell m}^M
	\end{align}
where the first term on the right-hand side of Eq. \eqref{eq:alm_inj} is given by Eq. \eqref{eq:p0} while the index
$f_b$ and $s$ correspond to the discrete frequency bands and the segments used to divide the Advanced LIGO-Virgo data, both of which are as described in Sec. \ref{sec:GW}. The number of data segments $S=15$ was chosen to maximize the number of samples to provide our statistical framework while also ensuring that each segment includes data from each part of the sidereal day. The second term on the right-hand side of Eq. \eqref{eq:alm_inj} is derived from the theoretical model and denoted with the superscript $M$. Note that the theoretical model does not directly provide values of ${{a}}_{\fb,\ell m}^M$; nevertheless, we can draw a set of ${a}_{\fb,\ell m}^M$ using the given model ${A}_\ell(f;\bmtheta)$. More specifically,
a given map element ${{a}}_{\fb,\ell m}^M$ is complex-valued. We assume its real and imaginary parts are independent random variables, each drawn from the multivariate normal distribution $\mathcal{N} (0,\frac{1}{2}A_\ell (f;\bmtheta))$. In the case when $m=0$, each ${{a}}_{\fb,\ell 0}^M$ is purely real and drawn instead from $\mathcal{N} (0,A_\ell (f;\bmtheta))$. The frequency values at which we evaluate the theoretical model $A_\ell (f;\bmtheta)$ are taken to be the midpoints, $\bar{\fb}$, of the frequency bands considered. 

We note that we sample the map only once and must treat that sample as though it is the only one provided by our realization of the universe. To make this sample map consistent with the angular power spectrum from which it was drawn, we introduce a scaling factor ${\tau}_{\fb,\ell}$:
\begin{align}\label{eq:r}
    \tau_{\fb,\ell} = \sqrt{ \frac{A_\ell(\bar{\fb};\bmtheta)}{\frac{1}{2 \ell + 1}\sum^{\ell}_{m=-\ell}|a^M_{\fb,\ell m}|^2}}.
\end{align}
We apply the scaling factor to the drawn map elements such that they satisfy the condition below:
\begin{align}\label{eq:A_ell}
		A_\ell (\bar{\fb};\bmtheta)=\frac{1}{1+2\ell}\sum_{m=-\ell}^{\ell} |\tau_{\fb,\ell}{a}^M_{\fb,\ell m}|^2.
	\end{align}
This coincides with the hypothesis testing problem \eqref{eqn:test} in Sec. \ref{sec:III}.

We follow \cite{yang2023} and apply the scaling factor $\mathcal{K}(f)$ to properly scale the LIGO results $\hat{a}_{\fb,\ell m}$ relative to the values provided by the theoretical model ${a}_{\fb,\ell m}^M$:
	\begin{align}\label{eq:K}
		\mathcal{K}(f)=\frac{2 \pi^2}{3 H_0^2} f_{\rm ref}^3 \left(\frac{f}{f_{\rm ref}}\right)^\alpha
	\end{align}
Here, $f_{\rm ref}$ is the reference frequency of 25Hz used in O3 \cite{O3stochdir}, and $f$ again takes on values of the midpoints of each frequency band $\bar{\fb}$. We also apply the square of this scaling factor to the inverted Fisher matrix in order to obtain a properly scaled covariance matrix for the $A_\ell$ (Eq. \eqref{eqn:cov_text}). For notational simplicity, for the remainder of this section, when we write ${a}_{\fb,\ell m}^M$, we assume it to be the version scaled by $\tau_{\fb,\ell}$, and when we write $\hat{a}_{\fb,s,\ell m}$ or $(\Gamma^{-1}_{\fb,s})_{\ell m, \ell'm'}$, we assume them to be the versions scaled by $\mathcal{K}(f)$ and $\mathcal{K}(f)^2$ respectively.

With the model map elements ${{a}}_{\fb,\ell m}^M$ and data map elements $\hat{{a}}_{\fb,s,\ell m}$ properly scaled, we may find their sum ${\hat{a}}_{\fb,s,\ell m}^{\rm inj}$ to complete the signal injection. The corresponding angular power spectrum ${\hat{A}}_{\fb,s,\ell}^{\rm inj}$ is calculated according to Eq. \eqref{eq:C_ell},
\begin{align}\label{eq:A_inj}
        \hat{A}^{\rm inj}_{\fb,s,\ell}=\frac{1}{1+2\ell}\sum_{m=-\ell}^{\ell} \left[|\hat{a}^{\rm inj}_{\fb,s,\ell m}|^2 - (\Gamma^{-1}_{\fb,s})_{\ell m,\ell m}\right],
    \end{align}
and the covariance matrix is computed according to Eq. \eqref{eqn:cov_text}, with ${a}_{\fb,\ell m}$ replaced by $\hat{a}^{\rm inj}_{\fb,\ell m}$, which is defined as the average of $\hat{a}^{\rm inj}_{\fb,s,\ell m}$ over the data segments $s$.


The necessary inputs for the procedure described in Sec. \ref{sec:III} are the data ${\hat{A}}_{\fb,s,\ell}^{\rm inj}$, their corresponding covariance matrices, and model values ${A}_\ell(\bar{\fb};\bmtheta)$. In this study, we test our procedure on two variations of these quantities. 

We compare the model described by Eqs.~\eqref{eq:modelCl}-\eqref{eq:A_z_f} to data injected using the same model, applying the covariance matrix as described above. When generating the injected signal, the parameter values are  $1.2 \times 10^{-28}$, $0.6$, and $0.7$ for $A_{max}$, $z_c$, and $\sigma_z$ respectively. The parameter estimates obtained via Eq~\eqref{optim} are $1.19 \times 10^{-28}$, $0.590$, and $0.694$, {\color{black}with corresponding standard errors $5.30 \times 10^{-31}$, $0.014$, and $0.012$.} Next, we use the same covariance matrix and again use a postulated model as given in Eq. \eqref{eq:A_z_f}, but in this case we use data which has been injected with a different variation of the model. In particular, we replace the astrophysical kernel with 
\begin{align}\label{eq:A_z_f_linear}
\mathcal{B}(f, z; \bmtheta) = f^{-1/3}(\beta_0+\beta_1z),
\end{align}
where $\bmtheta = (\beta_0, \beta_1)$. For this alternative model, we set $\beta_0 = 1.2 \times 10^{-28}$ and $\beta_1=6.0 \times 10^{-29}$ when generating the injected data. The parameter estimates for the postulated model are then $2.84 \times 10^{-28}$, $3.58$, and $2.72$ for $A_{max}$, $z_c$, and $\sigma_z$, respectively, {\color{black}with corresponding standard errors $8.29 \times 10^{-30}$, $0.138$, and $0.060$.}


To perform the test in Eq. \eqref{eqn:test}, we employ the test statistics $\widetilde{K}_N$ and $\widetilde{C}_N$ in Eq.~\eqref{eqn:TransTS}. We simulate their limiting distributions by running $10^5$ Monte Carlo simulations of the limiting process in Eq.~\eqref{eqn:limit_process_transformed}. The final results align with our expectations. Specifically, under the first scenario, $\widetilde{K}_N$ and $\widetilde{C}_N$ lead to p-values of 0.111 and 0.264, respectively, implying that the model is correctly identified as consistent with the data. Under the second scenario, however, the test statistics yield p-values 0.002 and 0.008, respectively, and correctly identify the model as not being consistent with the data. 

\section{Summary and Discussion}
This paper extends the distribution-free statistical framework presented in a companion paper \cite{PRL} to account for the unique challenges arising when estimating and testing models for the angular power spectrum of the SGWB. 

Framing the problem in the context of regression avoids the need to specify a likelihood function for the angular power spectrum. The procedure only requires a model for its mean and covariance matrix. To ensure consistency in our analysis, a new closed-form expression for the latter is derived to account for the variability associated with the maximum likelihood estimation of the spherical harmonic coefficients of the SGWB sky map. 

Simulation studies and an application to data from the third observing run of Advanced LIGO and Advanced Virgo demonstrated that the proposed testing procedure can detect departures from the hypothesized angular power spectral models when misspecified while adequately controlling for the probability of false discoveries. We note that the current sensitivity of the Advanced LIGO and Advanced Virgo detectors has not allowed for direct detection of the SGWB. To circumvent this limitation, we have used the O3 dataset as a noise realization, and a simulated signal has been injected into it. Given the good performance in such a setting, we expect the proposed framework to be valuable when applied to data collected by more sensitive detectors.

\textcolor{black}{In an astrophysical SGWB comprised of CBC events, the temporal discreteness of the events gives rise to a \textit{shot noise} which dominates the background signal \cite{Jenkins_2019_2, Jenkins_2019_shot}. In the application of our statistical framework to O3 data, we do not account for this noise, but alternative estimation methods for the angular power spectrum have been developed to obtain estimates that are unbiased by temporal shot noise \cite{Jenkins_2019_shot, Kouvatsos:2023bgd}. These methods involve subdividing the data set into separate segments as we have done in this study. In our framework, we perform this subdivision to obtain multiple sets of estimators to increase our sample size, but the aforementioned unbiased methods require this subdivision to obtain only a single set of estimators. It may therefore be difficult to account for the temporal shot noise while also satisfying our statistical framework's sample size requirements, and future work should be done to find a balance between these two factors.}

Finally, we note that inconsistencies in the statistical formalism can be found in the literature concerning the statistical analysis of the SGWB. In particular, when the spherical harmonic coefficients of the SGWB sky map are estimated via maximum likelihood \cite[cf.][]{Thrane:2009fp} the user implicitly assumes such quantities -- as well as the corresponding sky maps--  are deterministic. Therefore, they should not be simultaneously treated as zero-mean complex Gaussian random variables. A statistical analysis that enables the estimation and testing of angular power spectral models while also allowing the spherical harmonic coefficients to be random can be established. It requires, however, a revision of the modeling framework at a fundamental level. While this is beyond the scope of this manuscript, it is the subject of ongoing work.

\begin{acknowledgments}
SA and XZ are grateful for the financial support provided
by the Office of the Vice President for Research \& Innovation at the University of Minnesota. XZ was partially supported by the University of Minnesota Data Science Initiative with funding made available by the MnDrive initiative. SA, GJ, and VM were partially supported by NSF grant DMS-2152746.
The work of EF and VM was in part supported by NSF grants PHY-2110238 and PHY-1806630. \textcolor{black}{HZ was supported by the IRSA Faragher Distinguished Postdoctoral Fellowship.} The authors are grateful for computational resources provided by the LIGO Laboratory and supported by NSF Grants PHY-0757058 and PHY-0823459. This material is based upon work supported by NSF's LIGO Laboratory, which is a major facility fully funded by the National Science Foundation.

This research has made use of data or software obtained from the Gravitational Wave Open Science Center (gwosc.org), a service of the LIGO Scientific Collaboration, the Virgo Collaboration, and KAGRA. This material is based upon work supported by NSF's LIGO Laboratory which is a major facility fully funded by the National Science Foundation, as well as the Science and Technology Facilities Council (STFC) of the United Kingdom, the Max-Planck-Society (MPS), and the State of Niedersachsen/Germany for support of the construction of Advanced LIGO and construction and operation of the GEO600 detector. Additional support for Advanced LIGO was provided by the Australian Research Council. Virgo is funded, through the European Gravitational Observatory (EGO), by the French Centre National de Recherche Scientifique (CNRS), the Italian Istituto Nazionale di Fisica Nucleare (INFN) and the Dutch Nikhef, with contributions by institutions from Belgium, Germany, Greece, Hungary, Ireland, Japan, Monaco, Poland, Portugal, Spain. KAGRA is supported by Ministry of Education, Culture, Sports, Science and Technology (MEXT), Japan Society for the Promotion of Science (JSPS) in Japan; National Research Foundation (NRF) and Ministry of Science and ICT (MSIT) in Korea; Academia Sinica (AS) and National Science and Technology Council (NSTC) in Taiwan.
\end{acknowledgments}

\section*{Code Availability}
The Python code used to conduct the simulations and analyses in Sections \ref{sec:inference}, \ref{simulations}, and \ref{sec:application}, along with a tutorial for implementing the distribution-free tests in Python, is available at \url{https://github.com/xiangyu2022/DisfreeTestAPS}. The O3 data sets for LIGO and Virgo may be accessed via \url{https://gwosc.org/O3} \cite{Abbott_2023_data}. All codes necessary for performing the spherical harmonics decomposition search on the O3 data set are available from the public {\tt Stochastic} repository \url{https://git.ligo.org/stochastic-public/stochastic}. Values of the theoretical models of the SGWB angular power spectrum were generated using the {\tt CLASS} \url{https://class-code.net} \cite{Diego_Blas_2011} and {\tt MontePython} \cite{Benjamin_Audren_2013, brinckmann2018montepython3boostedmcmc} packages. The \textsf{R} package \textit{distfreereg} \citep[][]{Miller2024}, available at \url{https://cran.r-project.org/package=distfreereg}, contains a general implementation of the goodness-of-fit test. A tutorial for implementing the test in \textsf{R} is available at \url{https://github.com/small-epsilon/GOF-Testing-for-Angular-Power-Spectrum-Models-in-R}.

\appendix

\section{The generalized \texorpdfstring{$\chi^2$} ~~distribution of the angular power spectrum}
\label{generalized_chi}
 
Define the index set $\mathcal{I}_\ell=\{(\ell,m)\in \mathbb{N}^+ \times \mathbb{Z}; -\ell\leq  m\leq \ell\}$ and define
\begin{equation}
    \mathcal{I}=\bigcup_{\ell=1}^{\ell_{\rm max}} \mathcal{I}_\ell.
\end{equation}
From Eqs.~\eqref{eq:p0}-\eqref{eq:fisher} the vector of elements $\hat{a}_{\fb ,\ell m}$ is a linear combination of the complex Gaussian distributed random variable $C(t;f)$ with density defined as in Eq.~\eqref{equ:gauss_density}. {To ease notation, we write $\hat{a}_{\ell m}$ to denote $\hat{a}_{\fb, \ell m}$, ignoring the influence of $\fb$.} Therefore, $\{\hat{a}_{
\ell m}\}_{\mathcal{I}}$ is distributed as a multivariate complex Gaussian random variable with mean $\{  {a}_{\ell m}\}_{\mathcal{I}}$ and covariance matrix $\boldsymbol{\Gamma}^{-1}$. That is, 
\begin{equation}
    \{\hat{a}_{\ell m}\}_{\mathcal{I} } \sim \mathcal{CN} (\{ {a}_{\ell m}\}_{\mathcal{I} }, \boldsymbol{\Gamma}^{-1}),
\end{equation}
where $\boldsymbol{\Gamma}^{-1}= \{\Gamma^{-1}_{\ell m,\ell'm'}\}_{(\ell,m),(\ell',m')\in \mathcal{I} }$. 
Define
\begin{equation}
\tilde{\chi}_\ell=\sum_{m=-\ell}^\ell |\hat{a}_{\ell m}|^2.
\end{equation}
The goal is to show that $\tilde{\chi}_\ell$ has a generalized $\chi^2$ distribution.

Let $\boldsymbol{X}$ be the random vector whose components are the real and imaginary parts of $\{\hat{a}_{\ell m}\}_{\mathcal{I}}$. Because the density of $\{\hat{a}_{\ell m}\}_{\mathcal{I}}$ is proportional to
\begin{align}
    &\exp\{- (\hat{a}_{\ell' m'}-{a}_{\ell' m'})^* \Gamma_{\ell' m',\ell m} (\hat{a}_{\ell m}-{a}_{\ell' m'})\}\\
    =&\exp\{-\tfrac{1}{2}(\boldsymbol{X}-\boldsymbol{\mu})^T \boldsymbol{\Omega}^{-1}(\boldsymbol{X}-\boldsymbol{\mu} )\},
\end{align}
where 
\begin{equation}
    \boldsymbol{X}=
    \begin{pmatrix}
        \Rel(\{\hat{a}_{\ell m} \}_{\mathcal{I}}) \\
        \Img(\{\hat{a}_{\ell m} \}_{\mathcal{I}})\\
    \end{pmatrix},\quad \boldsymbol{\mu} =\begin{pmatrix}
        \Rel(\{ {a}_{\ell m} \}_{\mathcal{I} }) \\
        \Img(\{ {a}_{\ell m} \}_{\mathcal{I} })\\
    \end{pmatrix} 
\end{equation}
and
\begin{equation}
    \boldsymbol{\Omega}^{-1}= {2}\begin{pmatrix}
\Rel(\boldsymbol{\Gamma} ) & -\Img(\boldsymbol{\Gamma} ) \\
\Img(\boldsymbol{\Gamma} ) & \Rel(\boldsymbol{\Gamma} )
\end{pmatrix},
\end{equation}

it follows that $\boldsymbol{X}\sim \mathcal{N} \left( \boldsymbol{\mu}, \boldsymbol{\Omega} \right)$, i.e., a real multivariate normal with mean $\boldsymbol{\mu}$ and covariance $\boldsymbol{\Omega}$.
Because $\boldsymbol{\Gamma}=\Rel (\boldsymbol{\Gamma})+ i\Img(\boldsymbol{\Gamma})$ and $\boldsymbol{\Gamma}^{-1}=\Rel (\boldsymbol{\Gamma}^{-1})+ i\Img(\boldsymbol{\Gamma}^{-1})$ are inverse of each other, we have 
\begin{equation}
\boldsymbol{\Omega}=\frac{1}{2}\begin{pmatrix}
\Rel(\boldsymbol{\Gamma}^{-1}) & -\Img(\boldsymbol{\Gamma}^{-1}) \\
\Img(\boldsymbol{\Gamma}^{-1}) & \Rel(\boldsymbol{\Gamma}^{-1})
\end{pmatrix} .
\end{equation}
Subvectors of $\boldsymbol{X}$ also follow a (real) multivariate Gaussian distribution. In particular,
\begin{equation}
    \boldsymbol{X}_\ell =
    \begin{pmatrix}
        \Rel(\{\hat{a}_{\ell m} \}_{\mathcal{I}_\ell}) \\
        \Img(\{\hat{a}_{\ell m} \}_{\mathcal{I}_\ell})\\
    \end{pmatrix}\sim
    \mathcal{N} \left( \boldsymbol{\mu}_\ell, \boldsymbol{\Omega}_{\ell} \right),
\end{equation}
with mean
\begin{equation}
    \boldsymbol{\mu}_\ell=\begin{pmatrix}
        \Rel(\{ {a}_{\ell m} \}_{\mathcal{I}_\ell}) \\
        \Img(\{ {a}_{\ell m} \}_{\mathcal{I}_\ell})\\
    \end{pmatrix},
\end{equation}
    and covariance matrix
    \begin{equation}
    \boldsymbol{\Omega}_{\ell}=\frac{1}{2}\begin{pmatrix}
\Rel(\boldsymbol{\Gamma}^{-1})_{\ell,\ell} & -\Img(\boldsymbol{\Gamma}^{-1})_{\ell,\ell} \\
\Img(\boldsymbol{\Gamma}^{-1})_{\ell,\ell} & \Rel(\boldsymbol{\Gamma}^{-1})_{\ell,\ell}
\end{pmatrix}.
\end{equation}
Note that $\tilde{\chi}_\ell=\|\boldsymbol{X}_\ell\|^2$.
Since $\boldsymbol{\Gamma}$ is a positive definite matrix, $(\boldsymbol{\Gamma}^{-1})_{\ell,\ell}$ is also positive definite. 
Therefore, if, for any non-zero vector $(\boldsymbol{x}^T, \boldsymbol{y}^T )^T\in \mathbb{R}^{2l+1}\times \mathbb{R}^{2l+1}$, $\boldsymbol{z}=\boldsymbol{x}+ i\boldsymbol{y}$ and let $\boldsymbol{z}^\dagger$ be its conjugate transpose, then
\begin{equation}
\begin{split}
&\frac{1}{2}(\boldsymbol{x}^T, \boldsymbol{y}^T )\begin{pmatrix}
\Rel(\boldsymbol{\Gamma}^{-1})_{\ell,\ell} & -\Img(\boldsymbol{\Gamma}^{-1})_{\ell,\ell} \\
\Img(\boldsymbol{\Gamma}^{-1})_{\ell,\ell} & \Rel(\boldsymbol{\Gamma}^{-1})_{\ell,\ell}
\end{pmatrix}
\begin{pmatrix}
    \boldsymbol{x}\\
    \boldsymbol{y}
\end{pmatrix}\\
\quad&=\frac{1}{2}\boldsymbol{z}^\dagger  (\boldsymbol{\Gamma}^{-1})_{\ell,\ell} \boldsymbol{z}>0,    
\end{split}
\end{equation}
which implies that $\boldsymbol{\Omega}_\ell$ is also positive definite. 

Consider the eigendecomposition of $\boldsymbol{\Omega}_\ell$,
\begin{equation}
\boldsymbol{\Omega}_\ell=\boldsymbol{C}_\ell\boldsymbol{\Lambda}_\ell\boldsymbol{C}_\ell^T=\sum_{i=1}^{r_\ell}\lambda_{\ell,i} \boldsymbol{C}_{\ell,i}\boldsymbol{C}_{\ell,i}^T,
\end{equation}
where $\boldsymbol{C}_\ell=[\boldsymbol{C}_{\ell,1},\cdots,\boldsymbol{C}_{\ell,r_\ell}]$ is an orthogonal matrix, and $\boldsymbol{\Lambda}_\ell=\operatorname{diag}(\lambda_{\ell,1} \bm{I}_{k_{\ell,1}},\cdots, \lambda_{\ell,r_\ell} \bm{I}_{k_{\ell,r_\ell}})$ is a diagonal matrix. Here, \( r_\ell \) represents the number of distinct eigenvalues of the matrix \( \boldsymbol{\Omega}_\ell \), and \( \{ \lambda_{\ell,1}, \ldots, \lambda_{\ell,r_\ell} \} \) is the set of these distinct eigenvalues. The values \( k_{\ell,1}, \ldots, k_{\ell,r_\ell} \) represent the algebraic multiplicities of the corresponding eigenvalues \( \lambda_{\ell,1}, \ldots, \lambda_{\ell,r_\ell} \) for the matrix \( \boldsymbol{\Omega}_\ell \).

Set 
\begin{equation}
\begin{split}
\boldsymbol{U}_\ell&= \boldsymbol{\Lambda}_\ell^{-1/2}\boldsymbol{C}_\ell^T (\boldsymbol{X}_\ell-\boldsymbol{\mu}_\ell)=\begin{pmatrix}
\boldsymbol{U}_{\ell,1} \\
\boldsymbol{U}_{\ell,2} \\
\vdots \\
\boldsymbol{U}_{\ell,r}
\end{pmatrix},
\end{split}
\end{equation}
with $\boldsymbol{U}_{\ell,i}\in \mathbb{R}^{k_{\ell,i}}$ and
\begin{equation}
\begin{split}
\boldsymbol{\mu}_\ell'&= \boldsymbol{\Lambda}_\ell^{-1/2}\boldsymbol{C}_\ell^T\boldsymbol{\mu}_\ell=\begin{pmatrix}
\boldsymbol{\mu}'_{\ell,1} \\
\boldsymbol{\mu}'_{\ell,2} \\
\vdots \\
\boldsymbol{\mu}'_{\ell,r}
\end{pmatrix},
\end{split}
\end{equation}
where 
$ \boldsymbol{\mu}'_{\ell,i}=\lambda_{\ell,i}^{-1/2} \boldsymbol{C}_{\ell,i}^T\boldsymbol{\mu}_\ell \in \mathbb{R}^{k_{\ell,i}},
$ 
for any $1\leq i\leq r_\ell$.
It is easy to verify that 
\begin{equation}
    \boldsymbol{U}_\ell\sim \mathcal{N}(\bm{0},\bm{I}_{4\ell+2})\text{ and }\boldsymbol{X}_\ell=\boldsymbol{C}_\ell\boldsymbol{\Lambda}_\ell^{1/2}(\boldsymbol{U}_\ell + \boldsymbol{\mu}'_\ell).
\end{equation}
Therefore, 
\begin{equation}
\begin{split}
    \|\boldsymbol{X}_\ell \|^2 = &\|\boldsymbol{C}_\ell \boldsymbol{\Lambda}_\ell^{1/2}(\boldsymbol{U}_\ell + \boldsymbol{\mu}'_\ell)  \|^2 \\
    =&\|  \boldsymbol{\Lambda}_\ell^{1/2}(\boldsymbol{U}_\ell + \boldsymbol{\mu}'_\ell)  \|^2 = \sum_{i=1}^{r_\ell} \lambda_{\ell,i}\| \boldsymbol{U}_{\ell,i}+\boldsymbol{\mu}_{\ell,i}'\|^2,
\end{split}
\end{equation}
and thus, due to the correspondence between $\tilde{\chi}_\ell$ and $\|\boldsymbol{X}_\ell \|^2$, it follows that $\tilde{\chi}_\ell$ and $\sum_{i=1}^{r_\ell} \lambda_{\ell,i}\| \boldsymbol{U}_{\ell,i}+\boldsymbol{\mu}_{\ell,i}' \|^2$ have the same distribution. Furthermore, $\| \boldsymbol{U}_{\ell,i}+\boldsymbol{\mu}_{\ell,i}' \|^2$ follows a noncentral chi-squared distribution, denoted by $\chi^2(\theta_{\ell,i}, k_{\ell,i})$, with  noncentrality parameter 
\begin{equation}
\theta_{\ell,i} = \|\boldsymbol{\mu}'_{\ell,i}\|^2=\lambda_{\ell,i}^{-1}\boldsymbol{\mu}_\ell^T \boldsymbol{C}_{\ell,i}\boldsymbol{C}_{\ell,i}^T\boldsymbol{\mu}_\ell
\end{equation}
and $k_{\ell,i}$ degrees of freedom.

It follows that the random variable $\tilde{\chi}_\ell$ can be represented as a sum of independent noncentral chi-squared random variables
\begin{equation}
\label{eqn:sum_chi}
    \sum_{i=1}^{r_\ell} \lambda_{\ell,i}\cdot \chi^2(\theta_{\ell,i}, k_{\ell,i}).
\end{equation}

In general, the density of Eq.~\eqref{eqn:sum_chi} does not have a closed-form expression since it involves modified Bessel functions of the first kind \citep{kotz:bala:john:2000}. 
Because integrating the modified Bessel functions in the convolution lacks a simple closed-form solution, we cannot derive the density of the generalized chi-squared distribution unless 
$r_\ell=1$, that is, when no convolution and integration are needed. This situation occurs only when
$\boldsymbol{\Omega}_{\ell}=c\bm{I}_{4\ell+2}$ for some constant $c > 0$, which is equivalent to  
$(\boldsymbol{\Gamma}^{-1})_{\ell,\ell}=c'\bm{I}_{2\ell+1}$ for some constant $c' > 0$. 

\section{The covariance of the angular power spectrum}
\label{covariance_appendix}

The goal of this section is to prove the following result about the covariance of the angular power spectrum. 
\begin{proposition} 
\label{prop:cov angular ps}
At a given frequency bin $\fb$, and time segment $s$,
\begin{equation}\label{cov:A_l}
\begin{aligned}
    &\operatorname{Cov}(\hat{A}_{\fb,s,\ell},\hat{A}_{\fb,s,\ell'}) \\
    &=\sum_{m,m'}\frac{ |( {\Gamma^{-1}_{\fb,s}})_{\ell m,\ell' m'}|^2     + 2\Rel[a_{\fb,\ell m}^* ( {\Gamma}_{\fb,s}^{-1})_{\ell m,\ell' m'}a_{\fb,\ell' m'}]}{(1+2\ell)(1+2\ell')}.
\end{aligned}
\end{equation}
\end{proposition}

The proof of Proposition~\ref{prop:cov angular ps} is deferred until two preliminary lemmas have been established.

\begin{lemma}
\label{lem:cov bivariate normal}
If the random variables $U$ and $V$ satisfy
\begin{equation}
    \begin{pmatrix}
U \\
V
\end{pmatrix}\sim \mathcal{N}\left( \begin{pmatrix}
{\mu}_U \\
{\mu}_V
\end{pmatrix}, \begin{pmatrix}
\sigma_{11} & \sigma_{12} \\
\sigma_{12} & \sigma_{22}
\end{pmatrix} \right),
\end{equation}
then
\begin{equation}
\label{equ:UV}
    {\rm Cov} (U^2 ,V^2) = 2\sigma_{12}(2\mu_U\mu_V+\sigma_{12}). 
\end{equation}  
\end{lemma}

\begin{proof}
Notice that ${\rm Cov}(U^2, V^2) = \langle U^2 V^2 \rangle - \langle U^2\rangle \langle V^2 \rangle$ and
\begin{equation}
\label{eqn:Lemma1_cov3}  
    \langle U^2\rangle=\mu_U^2+\sigma_{11}, \text{ and } \langle V^2 \rangle =\mu_V^2+\sigma_{22}.
\end{equation}
All that is left is to calculate $\langle U^2 V^2 \rangle$.  The moment generating function of $(U,V)$, denoted as $M_{U,V}(a,b)$, is
\begin{equation}
\begin{split}
    M_{U,V}(a,b)&=\langle{e^{aU+bV}\rangle}\\
    &=\exp \left(\mu_U a+\mu_V b+\frac{1}{2}a^2\sigma_{11}+ab\sigma_{12}+\frac{1}{2}b^2\sigma_{22}\right)
\end{split}
\end{equation}
and direct calculation yields
\begin{equation}
\begin{split}
\label{eqn:Lemma1_cov2}    &\langle U^2V^2 \rangle =\frac{\partial^4}{\partial a^2 \partial b^2 } M_{U,V}(0,0)\\
    =&\mu_U^2 \mu_V^2 + \mu_V^2 \sigma_{11} + 4 \mu_U \mu_V \sigma_{12} + 2 \sigma_{12}^2 + \mu_U^2 \sigma_{22} + \sigma_{11} \sigma_{22}.
\end{split}
\end{equation}
Combining Eqs.~\eqref{eqn:Lemma1_cov3} and~\eqref{eqn:Lemma1_cov2} 
yields the result. 
\end{proof}

Consider a random variable $\mathbf{Z}$ that follows a multivariate complex Gaussian distribution with mean vector $\boldsymbol{\mu}\in \mathbb{C}^n$ 
 and covariance matrix $\boldsymbol{\Xi}$, which is Hermitian and non-negative definite, that is,
\begin{equation}
    \mathbf{Z} = [Z_1, Z_2, \cdots, Z_n]^T \sim \mathcal{CN}(\boldsymbol{\mu}, \boldsymbol{\Xi}).
\end{equation}
The real and imaginary parts of $\mathbf{Z}$ satisfy \citep[Cf.][for details]{Andersen1995}

\begin{equation} \label{eq:CN to MVN}
\begin{pmatrix}
\Rel(\mathbf{Z}) \\
\Img(\mathbf{Z})
\end{pmatrix}\sim \mathcal{N}\left( \begin{pmatrix}
\Rel(\boldsymbol{\mu}) \\
\Img(\boldsymbol{\mu})
\end{pmatrix} , \frac{1}{2}\begin{pmatrix}
\Rel(\boldsymbol{\Xi})& -\Img(\boldsymbol{\Xi}) \\
\Img(\boldsymbol{\Xi}) & \Rel(\boldsymbol{\Xi})
\end{pmatrix}
 \right).
\end{equation}

\begin{lemma} \label{lem:cov square} 
If  \( \mathbf{Z} = [Z_1, Z_2, \cdots, Z_n]^T \sim \mathcal{CN}(\boldsymbol{\mu}, \boldsymbol{\Xi}) \), then, for any indices $i,j=1,\dots,n$, 
\begin{equation}
\label{eqn:cov}
    {\rm Cov}(|{Z}_i|^2 ,|{Z}_j|^2) =|{\Xi}_{ij}|^2 + 2 \Rel( \mu_{Z_i}^*{\Xi}_{ij}\mu_{Z_j}),
\end{equation}
{where $\mu_{Z_i}$ denotes the $i$-th element of $\bm{\mu}$ and ${\Xi}_{ij}$ denotes the $(i,j)$-th element of the covariance matrix $\boldsymbol{\Xi}$.}
\end{lemma}

\begin{proof}

Without loss of generality consider $Z_1$ and $Z_2$.  Let $Z_1 = X_1 + i Y_1$ and $Z_2 = X_2 + i Y_2$, where $X_1$, $X_2$, $Y_1$, and $Y_2$ are real-valued Gaussian random variables. From Eq.~\eqref{eq:CN to MVN}, if $j,k \in \{1,2\}$, then $(X_j, Y_k)^T$ follows a bivariate Gaussian distribution. Let the means of $Z_j$, $X_j$, and $Y_j$ be denoted by $\mu_{Z_j}$, $\mu_{X_j}$, and $\mu_{Y_j}$, respectively.  Then $\mu_{Z_j} = \mu_{X_j} + i\mu_{Y_j}$. If $\sigma_{X_j,Y_k}$ represents the covariance between random variables $X_j$ and $Y_k$, then, using Lemma~\ref{lem:cov bivariate normal}, obtain
\begin{equation}
\begin{split}
\label{eqn:expansion}
    &{\rm Cov}(|{Z}_1|^2 ,|{Z}_2|^2)\\
    = & {\rm Cov}({X}_1^2 ,{X}_2^2) + {\rm Cov}({X}_1^2 ,{Y}_2^2) +{\rm Cov}({Y}_1^2 ,{X}_2^2) + {\rm Cov}({Y}_1^2 ,{Y}_2^2)\\
= & 2\sigma_{X_1,X_2}(2\mu_{X_1}\mu_{X_2} + \sigma_{X_1,X_2}) +2\sigma_{X_1,Y_2}(2\mu_{X_1}\mu_{Y_2} +\sigma_{X_1,Y_2})\\
& + 2\sigma_{Y_1,X_2}(2\mu_{Y_1}\mu_{X_2} + \sigma_{Y_1,X_2})+2\sigma_{Y_1,Y_2}(2\mu_{Y_1}\mu_{Y_2} +\sigma_{Y_1,Y_2}).
\end{split}
\end{equation}
If $\boldsymbol{\Xi}'$ is the covariance matrix between $Z_1$ and $Z_2$, that is, 
\begin{equation}
\begin{split}
    \boldsymbol{\Xi}' 
&= \begin{pmatrix} 
\Xi_{11} & \Xi_{12} \\
\Xi_{21} & \Xi_{22}
\end{pmatrix},
\end{split}
\end{equation}
then
\begin{equation}
    \quad \Rel(\boldsymbol{\Xi}') = \begin{pmatrix} 
\Xi_{11} & \Rel(\Xi_{12}) \\
\Rel(\Xi_{21})& \Xi_{22}\end{pmatrix},
\end{equation}
and
\begin{equation}
    \Img(\boldsymbol{\Xi}') = \begin{pmatrix} 
0 & -\Img(\Xi_{12}) \\
\Img(\Xi_{21})& 0\end{pmatrix}.
\end{equation}
From Eq.~\eqref{eq:CN to MVN}, the distribution of $(X_1,X_2, Y_1,Y_2)$ is 
\begin{equation} 
\label{eqn:cov1}
\begin{pmatrix}
X_1 \\
X_2 \\
Y_1 \\
Y_2
\end{pmatrix}\sim \mathcal{N}\left( \begin{pmatrix}
\mu_{X_1} \\
\mu_{X_2} \\
\mu_{Y_1}\\
\mu_{Y_2}\\
\end{pmatrix},   \boldsymbol{\Omega} \right),
\end{equation}
where
\begin{equation} \label{Omega}
     \boldsymbol{\Omega}=\frac{1}{2}
\begin{pmatrix}
\begin{matrix}
\Xi_{11} & \Rel(\Xi_{12}) \\
\Rel(\Xi_{21}) & \Xi_{22}  
\end{matrix} & \begin{matrix}
0 & -\Img(\Xi_{12}) \\
-\Img(\Xi_{21}) & 0
\end{matrix} \\
\begin{matrix}
0 & \Img(\Xi_{12}) \\
\Img(\Xi_{21}) & 0
\end{matrix}
& 
\begin{matrix}
\Xi_{11} & \Rel(\Xi_{12}) \\
\Rel(\Xi_{21}) & \Xi_{22}  
\end{matrix}
\end{pmatrix}. 
\end{equation}
Furthermore, since the covariance matrix $\boldsymbol{\Xi}$ is Hermitian, 
$\Rel(\boldsymbol{\Xi})=\Rel(\boldsymbol{\Xi})^T$, $\Img(\boldsymbol{\Xi})=-\Img(\boldsymbol{\Xi})^T$. Thus, \begin{equation}
\begin{split}
\label{eqn:equalities}
\sigma_{X_1,X_2}&=\sigma_{Y_1,Y_2}=\frac{1}{2}\Rel(\Xi_{12})=\frac{1}{2}\Rel(\Xi_{21}),\\
\sigma_{X_1,Y_2}&=-\frac{1}{2}\Img(\Xi_{12})=\frac{1}{2}\Img(\Xi_{21}),~~~\text{and}\\
\sigma_{Y_1,X_2}&=\frac{1}{2}\Img(\Xi_{12})=-\frac{1}{2}\Img(\Xi_{21}).
\end{split}
\end{equation}
By combining Eqs.~\eqref{eqn:expansion} and~\eqref{eqn:equalities} obtain
\begin{equation} 
\label{res}
\begin{split}
 &{\rm Cov}(|{Z}_1|^2 ,|{Z}_2|^2)\\
    =&(\Rel(\Xi_{21}))^2 +(\Img(\Xi_{21}))^2 + 2(\mu_{X_1}\mu_{Y_2} - \mu_{X_2}\mu_{Y_1}) \Img(\Xi_{21})\\
    & + 2(\mu_{X_1}\mu_{X_2}+\mu_{Y_1}\mu_{Y_2}) \Rel(\Xi_{21})\\
    =&|\Xi_{21}|^2 + 2 \Rel(\Xi_{21}) \Rel (\mu^*_{Z_1}\mu_{Z_2}) +2 \Img(\Xi_{21}) \Img (\mu^*_{Z_1}\mu_{Z_2})\\
    =&|\Xi_{12}|^2 + 2 \Rel( \mu_{Z_1}^*\Xi_{12}\mu_{Z_2}).
\end{split}
\end{equation}
The only remaining case to consider is when the indices coincide. Specifically, 
\begin{equation}
{\rm Cov}(|{Z}_1|^2 ,|{Z}_1|^2) = {\rm Var} (|{Z}_1|^2) = {\rm Var} (X_1^2) + {\rm Var}(Y_1^2). 
\end{equation}
The last equality follows since, by Eqs.~\eqref{eqn:cov1} and~\eqref{Omega}, the random variables $X_1$ and $Y_1$ are independent.   Also, note that from Eq.~\eqref{Omega}, ${\rm Var}(X_1) = {\rm Var}(Y_1) = \frac{1}{2} \Xi_{11}$. Using standard results about the moments of Gaussian distributions, direct calculation yields
\begin{equation}
\begin{split}
 {\rm Var}(X_1^2) = &\inner{X_1^4} - \inner{X_1^2}^2\\
    =&\frac{3}{4} \Xi_{11}^2 + 3\mu_{X_1}^2 \Xi_{11} + \mu_{X_1}^4 - \left( \frac{1}{2} \Xi_{11} + \mu_{X_1}^2 \right)^2\\
    =& \frac{1}{2} \Xi_{11}^2 + 2 \mu^2_{X_1} \Xi_{11}
    \end{split}
\end{equation}
and, similarly,
\begin{equation}
    {\rm Var}(Y_1^2) = \frac{1}{2} \Xi_{11}^2 + 2 \mu^2_{Y_1} \Xi_{11}.
\end{equation}
It follows that
\begin{equation}
\begin{split}
    {\rm Var}(|Z_1|^2)=& {\Xi}_{11}^2+2 |\mu_{Z_1}|^2  {\Xi}_{11} \\
    =& \Xi_{11}^2+2\Rel(\mu_{Z_1}^*\Xi_{11}  \mu_{Z_1}).
\end{split}
\end{equation}

\end{proof}

We are now in position to prove Proposition~\ref{prop:cov angular ps}.
\begin{proof}[Proof of Proposition~\ref{prop:cov angular ps}]
Since 
\begin{equation}
    \{\hat{a}_{\fb,s,\ell m}\}_{\mathcal{I} } \sim \mathcal{CN} (\{ {a}_{\fb,\ell m}\}_{\mathcal{I} }, \boldsymbol{\Gamma}_{\fb,s}^{-1}),
\end{equation}
where the index set $\mathcal{I}=\{(\ell,m)\in \mathbb{Z}^2;1\leq \ell\leq \ell_{\rm max},-\ell\leq m\leq \ell \}$, and $\boldsymbol{\Gamma}_{\fb,s}^{-1}= \{(\Gamma_{\fb,s}^{-1})_{\ell m,\ell'm'}  \}_{(\ell,m),(\ell',m')\in \mathcal{I} }$, applying Lemma~\ref{lem:cov square} yields

\begin{equation}
\begin{aligned}
& {\rm Cov}(|\hat{a}_{\fb,s,\ell m}|^2,|\hat{a}_{\fb,s,\ell'm'}|^2) \\
    &= |( {\Gamma}_{\fb,s}^{-1})_{\ell m,\ell' m'}|^2     + 2\Rel[a_{\fb,\ell m}^* ( {\Gamma}_{\fb,s}^{-1})_{\ell m,\ell' m'}a_{\fb,\ell' m'} ].
\end{aligned} 
\end{equation}

Recalling the definition of $\hat{A}_{\fb,s,\ell}$ in Eq.~\eqref{eq:C_ell} obtain
\begin{equation}
\begin{aligned}
  & {\rm Cov} (\hat{A}_{\fb,s,\ell},\hat{A}_{\fb,s,\ell'}) \\
& = \frac{1}{(1+2\ell)(1+2\ell')} \sum_{m,m'} {\rm Cov} (|\hat{a}_{\fb,s,\ell m}|^2,|\hat{a}_{\fb,s,\ell'm'}|^2),
\end{aligned}  
\end{equation}
and Eq.~\eqref{cov:A_l} follows.
\end{proof}

\section{Properties of the plug-in estimator of the covariance of angular power spectrum based on Eq. \texorpdfstring{\eqref{eqn:cov_text}}{(ref)}} 
\label{app:d}

Recall that at a given frequency band $\fb$, the unknown covariance function between \( \hat{A}_{\fb,s,\ell} \) and \(\hat{A}_{\fb,s,\ell'}\) is 
\begin{equation}
(\Sigma_{\fb,s} )_{\ell,\ell'} = {\rm Cov} (\hat{A}_{\fb,s,\ell}, \hat{A}_{\fb,s,\ell'}) .
\end{equation}
Assume that \( \{\hat{a}_{\fb,s,\ell m}\}_{\mathcal{I}}, \, s \geq 1 \), is a sequence of independent complex Gaussian random variables, with distribution \( \mathcal{CN}( \{ {a}_{\fb,\ell m}\}_{\mathcal{I}},  \boldsymbol{\Gamma}_{\fb,s}^{-1} ) \), where \( \boldsymbol{\Gamma}_{\fb,s} \) and \( \boldsymbol{\Gamma}^{-1}_{\fb,s} \) are deterministic and inverses of each other. The index set $\mathcal{I}$ is defined similarly as in Appendix~\ref{generalized_chi}.

As described in Section~\ref{sec:III.1}, the elements of the plug-in estimator of \( (\Sigma_{\fb,s} )_{\ell,\ell'} \) are given by
\begin{equation}
\begin{split}
     \sum_{m, m'} \frac{ |( {\Gamma}_{\fb,s}^{-1})_{\ell m, \ell' m'}|^2 + 2 \, \Rel \left[ \hat{a}^*_{\fb,\ell m} (  \Gamma_{\fb,s}^{-1})_{\ell m, \ell' m'} \hat{a}_{\fb,\ell' m'} \right]}{(1+2\ell)(1+2\ell')} .
\end{split}
\end{equation}
Recall from Eq. 
\eqref{equ:a-S}, $\hat{a}_{\fb,\ell m}$ depends on $S$.

Let $\{X_n\}_{n=1}^{\infty}$ be a sequence of random variables defined on the same probability space, and let $X$ be another random variable on this space. Say that $X_n$ converges in probability to $X$ as $n \to \infty$, denoted
\begin{equation}
  X_n \inP X.
\end{equation}


Let $\| \cdot \|_{\rm op}$ denote the operator norm, which is equal to the largest singular value of the matrix.

\begin{theorem} \label{thm5}
The following statements hold.
\begin{enumerate}
    \item Both \( \Sigma_{\fb,s} \) and \( \hat{\Sigma}_{\fb,s} \) are positive definite matrices.
    \item If there exists $\zeta >0$ such that, for all $s \ge 1$ and $\fb$, 
\begin{equation}\label{equ:condD1}
    \|\Gamma^{-1}_{\fb,s}\|_{\rm op}\leq \zeta,
\end{equation}
then $\hat{a}_{\fb,\ell m} \inP a_{\fb,\ell m}$, as $S \to \infty$.
    
\item If $\hat{a}_{\fb,\ell m}\inP {a}_{\fb,\ell m}$, then \( \hat{\Sigma}_{\fb,s} \) is a consistent estimator of \( \Sigma_{\fb,s} \), for any $s$.
    \item If condition~\eqref{equ:condD1} holds, then, as $S \to \infty$, 
\begin{equation}
     \max_{ \fb } \max_{1 \le s \le S}
    \bigl\|\hat{\Sigma}_{\fb,s}^{-1} - \Sigma_{\fb,s}^{-1}\bigr\|_{\rm op}
  \inP 0.
\end{equation}    
\end{enumerate}
\end{theorem}
\begin{proof}

\underline{Proof of statement 1:}
This proof will be the same for all frequency bins $\fb$ and segments $s$ considered. Therefore, for notational simplicity, let us denote  $\Sigma=\Sigma_{\fb,s}$, $\boldsymbol{\Gamma}=\Gamma_{\fb,s}$ and $a_{\ell m}=a_{\fb,\ell m}$. 

Let $d_1,\cdots,d_{\ell_{\max}}$ be complex numbers such that $|d_1|^2+\cdots+|d_{\ell_{\max}}|^2>0$. Let $\boldsymbol{d} = \{d_{\ell}/(2\ell+1)\}_{(\ell, m)\in \mathcal{I}}$. A direct calculation shows that 
\begin{equation}\label{equ:smallest_singular_value}
\begin{split}
\sum_{\ell,\ell'}d^*_{\ell} & (\Sigma  )_{\ell,\ell'} d_{\ell'}
    =\sum_{\ell m  ,\ell'm'}\frac{ d_{\ell}^* }{1+2\ell }  ( {\Gamma}^{-1})_{\ell' m',\ell m}( {\Gamma}^{-1})_{\ell m,\ell' m'}     \frac{ d_{\ell'} }{1+2\ell' } \\
    &+2\sum_{\ell m  ,\ell'm'}\frac{ d_{\ell}^* }{1+2\ell }    \Rel[a_{\ell m}^* ( {\Gamma}^{-1})_{\ell m,\ell' m'}a_{\ell' m'} ]\frac{ d_{\ell'} }{1+2\ell' } \\
    =&\operatorname{tr}\Big( (\boldsymbol{\Gamma}^{-1} \operatorname{diag}(\boldsymbol{d}))^*\boldsymbol{\Gamma}^{-1} \operatorname{diag}(\boldsymbol{d}) \Big) \\
    &+2\Rel \left( \sum_{\ell m  ,\ell'm'}  (\frac{ d_{\ell'} }{1+2\ell' } a_{\ell m})^* \big( {\Gamma}^{-1}\big)_{\ell m,\ell' m'}\big(a_{\ell' m'}   \frac{ d_{\ell} }{1+2\ell } \big) \right)\\
    \geq& \operatorname{tr}\Big( (\boldsymbol{\Gamma}^{-1} \operatorname{diag}(\boldsymbol{d}))^*\boldsymbol{\Gamma}^{-1} \operatorname{diag}(\boldsymbol{d}) \Big)>0.
\end{split}
\end{equation}
Thus, \( \Sigma \) is positive definite. The proof for \( \hat{\Sigma}_{\fb,s} \) is similar.

\noindent\underline{Proof of statement 2:} 
Recall that, as described in Section~\ref{sec:III}, $\Gamma_{\fb} = \sum_{s=1}^S \Gamma_{\fb,s}$. 
Hence, the smallest eigenvalue of $\Gamma_{\fb}$ is bounded below by $S / \zeta$. As $S \to \infty$, this lower bound goes to infinity; thus,
\begin{equation}
      \|\Gamma_{\fb}^{-1}\|_{\rm op} \;\to\; 0.
\end{equation}
Since 
\begin{equation}
      \{\hat{a}_{\fb,\ell m}\}_{\mathcal{I}}
  \;\sim\;
  \mathcal{CN}\bigl(\{a_{\fb,\ell m}\}_{\mathcal{I}}, \,\Gamma_{\fb}^{-1}\bigr),
\end{equation}
${\rm Cov} (\{\hat{a}_{\fb,\ell m}\}_{\mathcal{I}})=\Gamma_{\fb}^{-1}\to 0$ as $S\to \infty$, which implies that 
\begin{equation}
  \hat{a}_{\fb,\ell m} \;\inP\; a_{\fb,\ell m}
  \quad\text{as } S \to \infty.
\end{equation}

\noindent\underline{Proof of statement 3:} Define the continuous function $g_{\ell,\ell'}$ as 
{\small
\begin{equation}
\begin{split}
    g_{\ell,\ell'}(\{ \hat{a}_{f,\ell m} \}_{\mathcal{I}}, \boldsymbol{\Omega} )
    =\sum_{m  ,m'}\frac{   | {\Omega}_{\ell m,\ell' m'}|^2     + 2\Rel[\hat{a}_{f,\ell m}^*  {\Omega}_{\ell m,\ell' m'}\hat{a}_{f,\ell' m'} ]}{(1+2\ell)(1+2\ell')}.
\end{split}
\end{equation}
}



By the continuous mapping theorem \citep[Cf.][Ch. 2]{van2000asymptotic}, we obtain that as $\hat{a}_{\fb,\ell m} \inP {a}_{\fb,\ell m}$
\begin{equation}
    g_{\ell,\ell'}(\{ \hat{a}_{\fb,\ell m} \}_{\mathcal{I}},\Gamma^{-1}_{\fb,s}) \inP  g_{\ell,\ell'}(\{ { a}_{\fb,\ell m} \}_{\mathcal{I}},\Gamma^{-1}_{\fb,s}).
\end{equation}
Hence, \( \hat{\Sigma}_{\fb,s} \) is a consistent estimator of \( \Sigma_{\fb,s} \), as $\hat{a}_{\fb,\ell m}\inP {a}_{\fb,\ell m}$.


\noindent\underline{Proof of statement 4:}
Eq. \eqref{equ:smallest_singular_value} implies that there exists $\beta>0$ such that
\begin{equation}
\inf_{s\geq 1,\fb}\lambda_{min}(\Sigma_{\fb,s})\geq \beta,\text{ and }\inf_{s\geq 1,\fb}\lambda_{min}(\hat{\Sigma}_{\fb,s}) \geq \beta,
\end{equation}
where $\lambda_{min}(\Sigma_{\fb,s})$ denotes the smallest eigenvalue of $\Sigma_{\fb,s}$.
Since
\begin{equation}
\hat{\Sigma}^{-1}_{\fb,s}-{\Sigma}^{-1}_{\fb,s}= -\hat{\Sigma}^{-1}_{\fb,s}(\hat{\Sigma}_{\fb,s}-{\Sigma}_{\fb,s}){\Sigma}_{\fb,s}^{-1},
\end{equation}
we have
\begin{equation}
\begin{split}
    \|\hat{\Sigma}^{-1}_{\fb,s}-{\Sigma}^{-1}_{\fb,s}\|_{\rm op}&\leq \|\hat{\Sigma}_{\fb,s}^{-1}\|_{\rm op} \|\hat{\Sigma}_{\fb,s}-{\Sigma}_{\fb,s}\|_{\rm op} \|{\Sigma}_{\fb,s}^{-1}\|_{\rm op}\\
    &\leq \frac{1}{\beta^2}\|\hat{\Sigma}_{\fb,s}-{\Sigma}_{\fb,s}\|_{\rm op}.
\end{split}
\end{equation}
Note that there exists a positive constant $C$ such that
\begin{equation}
\begin{split}
    &|(\hat{\Sigma}_{\fb,s})_{\ell, \ell'}-({\Sigma}_{\fb,s})_{\ell, \ell'}|\\
    \leq& \sum_{m, m'} \frac{  2 \,   \left| ({a}^*_{\fb,\ell m}{a}_{\fb,\ell' m'} -\hat{a}^*_{\fb,\ell m}\hat{a}_{\fb,\ell' m'})  (  \Gamma_{\fb,s}^{-1})_{\ell m, \ell' m'}  \right|}{(1+2\ell)(1+2\ell')}\\
    \leq & {  C}\cdot \|\Gamma_{\fb,s}^{-1}\|_{\rm op} \sum_{m, m'} | {a}^*_{\fb,\ell m}{a}_{\fb,\ell' m'} -\hat{a}^*_{\fb,\ell m}\hat{a}_{\fb,\ell' m'} |\\
    \leq & {  C} \cdot \zeta \sum_{m, m'} | {a}^*_{\fb,\ell m}{a}_{\fb,\ell' m'} -\hat{a}^*_{\fb,\ell m}\hat{a}_{\fb,\ell' m'} |.    
\end{split}
\end{equation}
Since for any $\fb$ and any $|m|\leq \ell\leq \ell_{\rm max}$
\begin{equation}
      \hat{a}_{\fb,\ell m} \;\inP\; a_{\fb,\ell m}
  \quad\text{as } S \to \infty,
\end{equation}
there exists a positive constant $C'$ such that
\begin{equation}
\begin{aligned}
    & \max_{ \fb } \max_{1 \le s \le S}\|\hat{\Sigma}^{-1}_{\fb,s}-{\Sigma}^{-1}_{\fb,s}\|_{\rm op}\\
    \leq& \frac{\zeta\cdot C'}{\beta^2} \max_{ \fb } \sum_{m, m'} | {a}^*_{\fb,\ell m}{a}_{\fb,\ell' m'} -\hat{a}^*_{\fb,\ell m}\hat{a}_{\fb,\ell' m'} |\inP 0
\end{aligned}
\end{equation}
as $S\to \infty$.

\end{proof}

 

\section{Properties of the operator in Eq. \texorpdfstring{\eqref{Umu1r1}}{(Eq.~Umu1r1)} }
\label{app:e}
Consider the operator $U_{\bm{a},\bm{b}}$ on $\mathbb{R}^N$ introduced by \citet{khm13} and defined as
\begin{equation}
    \label{U}
    U_{\bm{a},\bm{b}} =\bm{I}-\frac{\langle\bm{a}-\bm{b}, \cdot \rangle}{1-\langle\bm{a},\bm{b}\rangle}(\bm{a}-\bm{b}),
\end{equation}
where $\bm{a}, \bm{b} \in \mathbb{R}^N$ satisfy
\begin{equation}
    ||\bm{a}||=||\bm{b}||=1, \quad \langle{\bm{a},\bm{b}\rangle}\neq 1.
\end{equation}
The operator $U_{\bm{a},\bm{b}}$ maps $\bm{a}$ to $\bm{b}$, maps $\bm{b}$ to $\bm{a}$, and keep the vectors that are orthogonal to both $\bm{a}$ and $\bm{b}$ unchanged, i.e., 
\begin{equation}
\begin{split}
U_{\bm{a},\bm{b}}\bm{a} &= \bm{b}, \  U_{\bm{a},\bm{b}}\bm{b} = \bm{a}, \\ 
\text{and}\quad U_{\bm{a},\bm{b}}\bm{x} &= \bm{x},\quad \text{for} \quad \bm{x} \perp \bm{a},\bm{b}.
\label{eqn:pro1}
\end{split}
\end{equation}

Such an operator $U_{\bm{a},\bm{b}}$ is unitary  -- that is, it is surjective and preserves the inner product. 

The surjectivity can be easily shown from Eq. \eqref{eqn:pro1}. Moreover, 
for any vectors $ \bm{x}, \bm{y} \in \mathbb{R}^N$:
\begin{equation}
\begin{aligned}
&\quad \ \langle U_{\bm{a},\bm{b}}\bm{x}, U_{\bm{a},\bm{b}}\bm{y}\rangle \\
&= \langle\bm{x} -\frac{\langle\bm{a}-\bm{b},\bm{x} \rangle}{1-\langle\bm{a},\bm{b}\rangle}(\bm{a}-\bm{b}), \bm{y} -\frac{\langle\bm{a}-\bm{b},\bm{y} \rangle}{1-\langle\bm{a},\bm{b}\rangle}(\bm{a}-\bm{b})\rangle \\ 
&= \langle \bm{x}, \bm{y} \rangle - \frac{\langle\bm{a}-\bm{b},\bm{x} \rangle \langle\bm{a}-\bm{b},\bm{y} \rangle }{(1-\langle\bm{a},\bm{b}\rangle)^2} \left(||\bm{a}-\bm{b}||^2-2-2\langle \bm{a}, \bm{b}\rangle\right) \\ 
&= \langle \bm{x}, \bm{y} \rangle, 
\end{aligned}
\end{equation}
therefore, $U_{\bm{a},\bm{b}}$ preserves the inner product. Moreover, $U_{\bm{a},\bm{b}}$ is a self-adjoint operator because it satisfies the condition 
\begin{equation}
\begin{aligned}
    \label{eqn:self-adjoint}
    \langle U_{\bm{a},\bm{b}}\bm{x}, \bm{y} \rangle 
&= \langle \bm{x}, \bm{y} \rangle  -  \frac{\langle\bm{a}-\bm{b},\bm{x} \rangle}{1-\langle\bm{a},\bm{b}\rangle}\langle\bm{a}-\bm{b},\bm{y}\rangle \\ 
&=  \langle\bm{x}, \bm{y} -\frac{\langle\bm{a}-\bm{b},\bm{y} \rangle}{1-\langle\bm{a},\bm{b}\rangle}(\bm{a}-\bm{b})\rangle \\
&= \langle \bm{x}, U_{\bm{a},\bm{b}}\bm{y} \rangle. 
\end{aligned}
\end{equation}


\pagebreak

\bibliography{PRD_new/biblio}

\begin{thebibliography}{66}
\providecommand{\natexlab}[1]{#1}
\providecommand{\url}[1]{\texttt{#1}}
\expandafter\ifx\csname urlstyle\endcsname\relax
  \providecommand{\doi}[1]{doi: #1}\else
  \providecommand{\doi}{doi: \begingroup \urlstyle{rm}\Url}\fi

\bibitem[Algeri et~al.(2025)Algeri, Zhang, Zhao, Miller, Floden, Jones, and
  Mandic]{PRL}
S.~Algeri, X.~Zhang, H.~Zhao, J.~Miller, E.~Floden, G.~L. Jones, and V.~Mandic.
\newblock A distribution-free approach to testing models for angular power
  spectra.
\newblock \emph{To be submitted to Phys. Rev. Letters}, 2025.

\bibitem[Collaboration et~al.(2020)Collaboration, Aghanim, et~al.]{planck2020}
Planck Collaboration, N.~Aghanim, et~al.
\newblock Planck 2018 results - vi. cosmological parameters.
\newblock \emph{Astron. Astrophys.}, 641:\penalty0 A6, 2020.

\bibitem[Ahumada et~al.(2020)]{SDSS_DR16}
R.~Ahumada et~al.
\newblock The 16th data release of the sloan digital sky surveys: First release
  from the apogee-2 southern survey and full release of eboss spectra.
\newblock \emph{Astrophys. J. Suppl. Ser.}, 249\penalty0 (1):\penalty0 3, 2020.

\bibitem[Jeffrey et~al.(2021)]{DESY3LensingMap}
N.~Jeffrey et~al.
\newblock Dark energy survey year 3 results: Curved-sky weak lensing mass map
  reconstruction.
\newblock \emph{Mon. Not. R. Astron. Soc.}, 505\penalty0 (3):\penalty0
  4626--4645, 2021.

\bibitem[Aasi et~al.(2015)]{Aasi_2015}
J.~Aasi et~al.
\newblock Advanced ligo.
\newblock \emph{Class. Quantum Grav.}, 32\penalty0 (7):\penalty0 074001, 2015.

\bibitem[Acernese et~al.(2014)]{Acernese_2015}
F.~Acernese et~al.
\newblock Advanced virgo: a second-generation interferometric gravitational
  wave detector.
\newblock \emph{Class. Quantum Grav.}, 32\penalty0 (2):\penalty0 024001, 2014.

\bibitem[Abbott et~al.(2021{\natexlab{a}})]{O3stochdir}
R.~Abbott et~al.
\newblock Search for anisotropic gravitational-wave backgrounds using data from
  advanced ligo and advanced virgo's first three observing runs.
\newblock \emph{Phys. Rev. D}, 104:\penalty0 022005, 2021{\natexlab{a}}.

\bibitem[Abbott et~al.(2021{\natexlab{b}})]{o3stochiso}
R.~Abbott et~al.
\newblock Upper limits on the isotropic gravitational-wave background from
  advanced ligo and advanced virgo’s third observing run.
\newblock \emph{Phys. Rev. D}, 104\penalty0 (2), 2021{\natexlab{b}}.

\bibitem[Collaboration et~al.(2025{\natexlab{a}})Collaboration, the
  Virgo~Collaboration, and the KAGRA~Collaboration]{o4isotropic}
The LIGO~Scientific Collaboration, the Virgo~Collaboration, and the
  KAGRA~Collaboration.
\newblock Upper limits on the isotropic gravitational-wave background from the
  first part of {L}{I}{G}{O}, {V}irgo, and {K}{A}{G}{R}{A}'s fourth observing
  run.
\newblock \emph{arXiv 2508.20721}, 2025{\natexlab{a}}.

\bibitem[Collaboration et~al.(2025{\natexlab{b}})Collaboration, the
  Virgo~Collaboration, and the KAGRA~Collaboration]{o4aniso}
The LIGO~Scientific Collaboration, the Virgo~Collaboration, and the
  KAGRA~Collaboration.
\newblock Directional search for persistent gravitational waves: Results from
  the first part of {L}{I}{G}{O}-{V}irgo-{K}{A}{G}{R}{A}'s fourth observing
  run.
\newblock \emph{arXiv 2510.17487}, 2025{\natexlab{b}}.

\bibitem[Contaldi(2017)]{Contaldi_2017}
C.~R. Contaldi.
\newblock Anisotropies of gravitational wave backgrounds: A line of sight
  approach.
\newblock \emph{Phys. Lett. B}, 771, 2017.

\bibitem[Jenkins et~al.(2019{\natexlab{a}})Jenkins, O'Shaughnessy,
  Sakellariadou, and Wysocki]{jenkins2019}
A.~C. Jenkins, R.~O'Shaughnessy, M.~Sakellariadou, and D.~Wysocki.
\newblock Anisotropies in the astrophysical gravitational-wave background: The
  impact of black hole distributions.
\newblock \emph{Phys. Rev. Lett.}, 122\penalty0 (11):\penalty0 111101,
  2019{\natexlab{a}}.

\bibitem[Jenkins and Sakellariadou(2019)]{Jenkins_2019_2}
A.~C. Jenkins and M.~Sakellariadou.
\newblock Shot noise in the astrophysical gravitational-wave background.
\newblock \emph{Phys. Rev. D}, 100\penalty0 (6), 2019.

\bibitem[Pitrou et~al.(2020)Pitrou, Cusin, and Uzan]{Pitrou_2020}
C.~Pitrou, G.~Cusin, and J.-P. Uzan.
\newblock Unified view of anisotropies in the astrophysical gravitational-wave
  background.
\newblock \emph{Phys. Rev. D}, 101\penalty0 (8), 2020.

\bibitem[Cusin et~al.(2017)Cusin, Pitrou, and Uzan]{Cusin:2017fwz}
G.~Cusin, C.~Pitrou, and J.-P. Uzan.
\newblock Anisotropy of the astrophysical gravitational wave background:
  Analytic expression of the angular power spectrum and correlation with
  cosmological observations.
\newblock \emph{Phys. Rev. D}, 96\penalty0 (10):\penalty0 103019, 2017.

\bibitem[Cusin et~al.(2018{\natexlab{a}})Cusin, Pitrou, and
  Uzan]{Cusin:2017mjm}
G.~Cusin, C.~Pitrou, and J.-P. Uzan.
\newblock The signal of the gravitational wave background and the angular
  correlation of its energy density.
\newblock \emph{Phys. Rev. D}, 97\penalty0 (12):\penalty0 123527,
  2018{\natexlab{a}}.

\bibitem[Cusin et~al.(2018{\natexlab{b}})Cusin, Dvorkin, Pitrou, and
  Uzan]{Cusin_2018_2}
G.~Cusin, I.~Dvorkin, C.~Pitrou, and J.-P. Uzan.
\newblock First predictions of the angular power spectrum of the astrophysical
  gravitational wave background.
\newblock \emph{Phys. Rev. Lett.}, 120\penalty0 (23), 2018{\natexlab{b}}.

\bibitem[Cusin et~al.(2019)Cusin, Dvorkin, Pitrou, and Uzan]{Cusin:2019jpv}
G.~Cusin, I.~Dvorkin, C.~Pitrou, and J.-P. Uzan.
\newblock Properties of the stochastic astrophysical gravitational wave
  background: astrophysical sources dependencies.
\newblock \emph{Phys. Rev. D}, 100\penalty0 (6):\penalty0 063004, 2019.

\bibitem[Geller et~al.(2018)Geller, Hook, Sundrum, and Tsai]{Geller}
M.~Geller, A.~Hook, R.~Sundrum, and Y.~Tsai.
\newblock Primordial anisotropies in the gravitational wave background from
  cosmological phase transitions.
\newblock \emph{Phys. Rev. Lett.}, 121:\penalty0 201303, 2018.

\bibitem[Bertacca et~al.(2020)Bertacca, Ricciardone, Bellomo, Jenkins,
  Matarrese, Raccanelli, Regimbau, and Sakellariadou]{Bertacca_2020}
D.~Bertacca, A.~Ricciardone, N.~Bellomo, A.~C. Jenkins, S.~Matarrese,
  A.~Raccanelli, T.~Regimbau, and M.~Sakellariadou.
\newblock Projection effects on the observed angular spectrum of the
  astrophysical stochastic gravitational wave background.
\newblock \emph{Phys. Rev. D}, 101\penalty0 (10), 2020.

\bibitem[Talukder et~al.(2014)Talukder, Thrane, Bose, and
  Regimbau]{Talukder_2014}
D.~Talukder, E.~Thrane, S.~Bose, and T.~Regimbau.
\newblock Measuring neutron-star ellipticity with measurements of the
  stochastic gravitational-wave background.
\newblock \emph{Phys. Rev. D}, 89:\penalty0 123008, 2014.

\bibitem[Mazumder et~al.(2014)Mazumder, Mitra, and Dhurandhar]{Mazumder_2014}
N.~Mazumder, S.~Mitra, and S.~Dhurandhar.
\newblock Astrophysical motivation for directed searches for a stochastic
  gravitational wave background.
\newblock \emph{Phys. Rev. D}, 89\penalty0 (8), 2014.

\bibitem[Jenkins and Sakellariadou(2018)]{Jenkins_2018}
A.~C. Jenkins and M.~Sakellariadou.
\newblock Anisotropies in the stochastic gravitational-wave background:
  Formalism and the cosmic string case.
\newblock \emph{Phys. Rev. D}, 98\penalty0 (6), 2018.

\bibitem[Allen and Romano(1999)]{allenromano}
B.~Allen and J.~D. Romano.
\newblock Detecting a stochastic background of gravitational radiation: Signal
  processing strategies and sensitivities.
\newblock \emph{Phys. Rev. D}, 59:\penalty0 102001, 1999.

\bibitem[Thrane et~al.(2009)Thrane, Ballmer, Romano, Mitra, Talukder, Bose, and
  Mandic]{Thrane:2009fp}
E.~Thrane, S.~Ballmer, J.~D. Romano, S.~Mitra, D.~Talukder, S.~Bose, and
  V.~Mandic.
\newblock Probing the anisotropies of a stochastic gravitational-wave
  background using a network of ground-based laser interferometers.
\newblock \emph{Phys. Rev. D}, 80:\penalty0 122002, 2009.

\bibitem[Tsukada et~al.(2023)Tsukada, Jaraba, Agarwal, and
  Floden]{Tsukada_2023}
L.~Tsukada, S.~Jaraba, D.~Agarwal, and E.~Floden.
\newblock Bayesian parameter estimation for targeted anisotropic
  gravitational-wave background.
\newblock \emph{Phys. Rev. D}, 107\penalty0 (2), 2023.

\bibitem[Agarwal et~al.(2022)Agarwal, Suresh, Mandic, Matas, and
  Regimbau]{Agarwal_MilkyWay}
D.~Agarwal, J.~Suresh, V.~Mandic, A.~Matas, and T.~Regimbau.
\newblock Targeted search for the stochastic gravitational-wave background from
  the galactic millisecond pulsar population.
\newblock \emph{Phys. Rev. D}, 106:\penalty0 043019, 2022.

\bibitem[Agarwal et~al.(2023{\natexlab{a}})Agarwal, Suresh, Mitra, and
  Ain]{Agarwal_2023}
D.~Agarwal, J.~Suresh, S.~Mitra, and A.~Ain.
\newblock Angular power spectra of anisotropic stochastic gravitational wave
  background: Developing statistical methods and analyzing data from
  ground-based detectors.
\newblock \emph{Phys. Rev. D}, 108\penalty0 (2), 2023{\natexlab{a}}.

\bibitem[Yang et~al.(2023)Yang, Suresh, Cusin, Banagiri, Feist, Mandic,
  Scarlata, and Michaloliakos]{yang2023}
K.~Z. Yang, J.~Suresh, G.~Cusin, S.~Banagiri, N.~Feist, V.~Mandic, C.~Scarlata,
  and I.~Michaloliakos.
\newblock Measurement of the cross-correlation angular power spectrum between
  the stochastic gravitational wave background and galaxy overdensity.
\newblock \emph{Phys. Rev. D}, 108:\penalty0 043025, 2023.

\bibitem[Romano and Cornish(2017)]{RomanoCornish}
J.~D. Romano and N.~J. Cornish.
\newblock Detection methods for stochastic gravitational-wave backgrounds: a
  unified treatment.
\newblock \emph{Living Rev. Rel.}, 20:\penalty0 2, 2017.

\bibitem[Liang et~al.(2024)Liang, Li, Li, Zhang, and Hu]{PhysRevD.110.043031}
Z.-C. Liang, Z.-Y. Li, E.-K. Li, J.-D. Zhang, and Y.-M. Hu.
\newblock Sensitivity to anisotropic gravitational-wave background with
  space-borne detector networks.
\newblock \emph{Physical Review D}, 110:\penalty0 043031, Aug 2024.

\bibitem[Abbott et~al.(2017{\natexlab{a}})]{o1iso}
B.~P. Abbott et~al.
\newblock Upper limits on the stochastic gravitational-wave background from
  advanced ligo's first observing run.
\newblock \emph{Phys. Rev. Lett.}, 118\penalty0 (12):\penalty0 121101,
  2017{\natexlab{a}}.

\bibitem[Abbott et~al.(2007)]{s4iso}
B.~Abbott et~al.
\newblock Searching for a stochastic background of gravitational waves with the
  laser interferometer gravitational-wave observatory.
\newblock \emph{Astrophys. J.}, 659:\penalty0 918--930, 2007.

\bibitem[Agarwal et~al.(2023{\natexlab{b}})Agarwal, Suresh, Mitra, and
  Ain]{PhysRevD.108.023011}
D.~Agarwal, J.~Suresh, S.~Mitra, and A.~Ain.
\newblock Angular power spectra of anisotropic stochastic gravitational wave
  background: Developing statistical methods and analyzing data from
  ground-based detectors.
\newblock \emph{Phys. Rev. D}, 108:\penalty0 023011, 2023{\natexlab{b}}.

\bibitem[Floden et~al.(2022)Floden, Mandic, Matas, and Tsukada]{floden2022}
E.~Floden, V.~Mandic, A.~Matas, and L.~Tsukada.
\newblock Angular resolution of the search for anisotropic stochastic
  gravitational-wave background with terrestrial gravitational-wave detectors.
\newblock \emph{Phys. Rev. D}, 106:\penalty0 023010, 2022.

\bibitem[Janssens(2023)]{Janssens:2023anf}
K.~Janssens.
\newblock {Prospects for an isotropic gravitational wave background detection
  with Earth-based interferometric detectors and the threat of correlated
  noise}.
\newblock In \emph{{57th Rencontres de Moriond on Gravitation}}, 5 2023.

\bibitem[Venikoudis et~al.(2025)Venikoudis, De~Lillo, Janssens, Suresh, and
  Bruno]{Venikoudis:2024zmk}
S.~Venikoudis, F.~De~Lillo, K.~Janssens, J.~Suresh, and G.~Bruno.
\newblock {Impact of correlated magnetic noise on directional stochastic
  gravitational-wave background searches}.
\newblock \emph{Physical Review D}, 111\penalty0 (8):\penalty0 082005, 2025.

\bibitem[Regimbau(2011)]{Regimbau_2011}
T.~Regimbau.
\newblock The astrophysical gravitational wave stochastic background.
\newblock \emph{Research in Astronomy and Astrophysics}, 11\penalty0
  (4):\penalty0 369–390, March 2011.

\bibitem[Abbott et~al.(2017{\natexlab{b}})Abbott, Abbott, Abbott, Abernathy,
  Acernese, Ackley, Adams, Adams, Addesso, Adhikari,
  et~al.]{abbott2017directional}
B.~P. Abbott, R.~Abbott, T.~D. Abbott, M.~R. Abernathy, F.~Acernese, K.~Ackley,
  C.~Adams, T.~Adams, P.~Addesso, R.~X. Adhikari, et~al.
\newblock Directional limits on persistent gravitational waves from advanced
  ligo’s first observing run.
\newblock \emph{Phys. Rev. Lett.}, 118\penalty0 (12):\penalty0 121102,
  2017{\natexlab{b}}.

\bibitem[Abbott et~al.(2019)Abbott, Abbott, Abbott, Abraham, Acernese, Ackley,
  Adams, Adhikari, Adya, Affeldt, et~al.]{abbott2019directional}
B.~P. Abbott, R.~Abbott, T.~D. Abbott, S.~Abraham, F.~Acernese, K.~Ackley,
  C.~Adams, R.~X. Adhikari, V.~B. Adya, C.~Affeldt, et~al.
\newblock Directional limits on persistent gravitational waves using data from
  advanced ligo’s first two observing runs.
\newblock \emph{Phys. Rev. D}, 100\penalty0 (6):\penalty0 062001, 2019.

\bibitem[Panda et~al.(2019)Panda, Bhagwat, Suresh, and
  Mitra]{PhysRevD.100.043541}
S.~Panda, S.~Bhagwat, J.~Suresh, and S.~Mitra.
\newblock Stochastic gravitational wave background mapmaking using regularized
  deconvolution.
\newblock \emph{Phys. Rev. D}, 100:\penalty0 043541, 2019.

\bibitem[Higham(2008)]{higham}
N.~J. Higham.
\newblock \emph{Functions of Matrices: Theory and Computation}.
\newblock SIAM, 2008.

\bibitem[Horn and Johnson(2012)]{Horn_Johnson}
R.~A. Horn and C.~R. Johnson.
\newblock \emph{Matrix Analysis}.
\newblock Cambridge University Press, 2nd edition, 2012.

\bibitem[Nocedal and Wright(1999)]{1999numerical}
J.~Nocedal and S.~J. Wright.
\newblock \emph{Numerical Optimization}.
\newblock Springer, 1999.

\bibitem[Davidson and MacKinnon(2004)]{davidson}
R.~Davidson and J.~G. MacKinnon.
\newblock \emph{Econometric Theory and Methods}.
\newblock Oxford University Press, 2004.

\bibitem[van~der Vaart(2000)]{van2000asymptotic}
A.~W. van~der Vaart.
\newblock \emph{Asymptotic Statistics}, volume~3.
\newblock Cambridge University Press, 2000.

\bibitem[Freedman(1981)]{freedman81}
D.~A. Freedman.
\newblock Bootstrapping regression models.
\newblock \emph{Ann. Statist.}, 9\penalty0 (6):\penalty0 1218--1228, 1981.

\bibitem[Khmaladze(1984)]{khm84}
E.~V. Khmaladze.
\newblock Martingale limit theorems for divisible statistics.
\newblock \emph{Theory Probab. Appl.}, 28\penalty0 (3):\penalty0 530--548,
  1984.

\bibitem[Algeri and Khmaladze(2024)]{algeri2024pearson}
S.~Algeri and E.~V. Khmaladze.
\newblock When {P}earson $\chi^2$ and other divisible statistics are not
  goodness-of-fit tests.
\newblock \emph{arXiv preprint arXiv:2406.09195}, 2024.

\bibitem[Acharya and Kashyap(2024)]{acharya2024spectral}
A.~Acharya and V.~L. Kashyap.
\newblock Spectral fit residuals as an indicator to increase model complexity.
\newblock \emph{Res. Notes Am. Astron. Soc.}, 8\penalty0 (1):\penalty0 1, 2024.

\bibitem[Loynes(1980)]{loynes80}
R.~M. Loynes.
\newblock The empirical distribution function of residuals from generalised
  regression.
\newblock \emph{Ann. Statist.}, 8\penalty0 (2):\penalty0 285--298, 1980.

\bibitem[Stute(1997)]{stute97}
W.~Stute.
\newblock Nonparametric model checks for regression.
\newblock \emph{Ann. Statist.}, pages 613--641, 1997.

\bibitem[Khmaladze(2021)]{khm21}
E.~V. Khmaladze.
\newblock Distribution-free testing in linear and parametric regression.
\newblock \emph{Ann. Inst. Stat. Math.}, 73\penalty0 (6):\penalty0 1063--1087,
  2021.

\bibitem[Philipp and Stout(1975)]{philipp}
W.~Philipp and W.~F. Stout.
\newblock \emph{Almost Sure Invariance Principles for Partial Sums of Weakly
  Dependent Random Variables}, volume 161.
\newblock American Mathematical Soc., 1975.

\bibitem[Bradley(2005)]{bradley2005basic}
R.~C. Bradley.
\newblock Basic properties of strong mixing conditions. a survey and some open
  questions.
\newblock \emph{Probab. Surv.}, 2:\penalty0 107--144, 2005.

\bibitem[Davis et~al.(2021)Davis, Areeda, et~al.]{Davis_2021}
D.~Davis, J.~S. Areeda, et~al.
\newblock Ligo detector characterization in the second and third observing
  runs.
\newblock \emph{Class. Quantum Grav.}, 38\penalty0 (13):\penalty0 135014, 2021.

\bibitem[Jenkins et~al.(2019{\natexlab{b}})Jenkins, Romano, and
  Sakellariadou]{Jenkins_2019_shot}
A.C. Jenkins, J.D. Romano, and M.~Sakellariadou.
\newblock Estimating the angular power spectrum of the gravitational-wave
  background in the presence of shot noise.
\newblock \emph{Physical Review D}, 100\penalty0 (8), October
  2019{\natexlab{b}}.

\bibitem[Kouvatsos et~al.(2024)Kouvatsos, Jenkins, Renzini, Romano, and
  Sakellariadou]{Kouvatsos:2023bgd}
N.~Kouvatsos, A.C. Jenkins, A.I. Renzini, J.D. Romano, and M.~Sakellariadou.
\newblock {Unbiased estimation of gravitational-wave anisotropies from noisy
  data}.
\newblock \emph{Physical Review D}, 109\penalty0 (10):\penalty0 103535, 2024.

\bibitem[R.~Abbott et~al. (LIGO Scientific~Collaboration and
  Collaboration)(2023)]{Abbott_2023_data}
Virgo~Collaboration R.~Abbott et~al. (LIGO Scientific~Collaboration and KAGRA
  Collaboration).
\newblock Open data from the third observing run of {LIGO}, {Virgo}, {KAGRA},
  and {GEO}.
\newblock \emph{Astrophys. J. Suppl. Ser.}, 267\penalty0 (2):\penalty0 29,
  2023.

\bibitem[Blas et~al.(2011)Blas, Lesgourgues, and Tram]{Diego_Blas_2011}
D.~Blas, J.~Lesgourgues, and T.~Tram.
\newblock The cosmic linear anisotropy solving system (class). part ii:
  Approximation schemes.
\newblock \emph{J. Cosmol. Astropart. Phys.}, 2011\penalty0 (07):\penalty0
  034–034, July 2011.

\bibitem[Audren et~al.(2013)Audren, Lesgourgues, Benabed, and
  Prunet]{Benjamin_Audren_2013}
B.~Audren, J.~Lesgourgues, K.~Benabed, and S.~Prunet.
\newblock Conservative constraints on early cosmology with {MONTE PYTHON}.
\newblock \emph{J. Cosmol. Astropart. Phys.}, 2013\penalty0 (02):\penalty0
  001–001, 2013.

\bibitem[Brinckmann and
  Lesgourgues(2018)]{brinckmann2018montepython3boostedmcmc}
T.~Brinckmann and J.~Lesgourgues.
\newblock {MontePython} 3: boosted {MCMC} sampler and other features, 2018.

\bibitem[Miller(2024)]{Miller2024}
J.~Miller.
\newblock {distfreereg}: Distribution-free goodness-of-fit testing for
  regression, November 2024.

\bibitem[Kotz et~al.(2000)Kotz, Balakrishnan, and Johnson]{kotz:bala:john:2000}
S.~Kotz, N.~Balakrishnan, and N.~L. Johnson.
\newblock \emph{Continuous Multivariate Distributions: Volume 1 Models and
  Applications}.
\newblock John Wiley \& Sons, Inc., New York, 2000.

\bibitem[Andersen et~al.(1995)Andersen, H{\o}jbjerre, S{\o}rensen, and
  Eriksen]{Andersen1995}
H.~H. Andersen, M.~H{\o}jbjerre, D.~S{\o}rensen, and P.~S. Eriksen.
\newblock \emph{Linear and Graphical Models for the Multivariate Complex Normal
  Distribution}, pages 15--37.
\newblock Springer New York, New York, NY, 1995.

\bibitem[Khmaladze(2013)]{khm13}
E.~Khmaladze.
\newblock Note on distribution free testing for discrete distributions.
\newblock \emph{Ann. Statist.}, pages 2979--2993, 2013.

\end{thebibliography}

\end{document}